\documentclass[11pt]{article}
\usepackage[T1]{fontenc}
\usepackage{fullpage,amssymb,graphicx}
\usepackage{xcolor}
\usepackage{framed}

\usepackage[
  pdfauthor={Frank Nielsen},
  pdftitle={Curved representational Bregman divergences and their applications},
  pdfsubject={arxiv},
  pdfkeywords={Bregman divergences ,curved exponential family, monotonic embeddings , $\alpha$-divergences ,centroids , Bregman projection, space of spheres, cosine dissimilarity , Bregman bias-variance decomposition}
]{hyperref}

\newenvironment{proof}{Proof:}{$\square$}

\DeclareSymbolFont{YHlargesymbols}{OMX}{yhex}{m}{n}
\DeclareMathAccent{\wideparen}{\mathord}{YHlargesymbols}{"F3}

\usepackage{comment,url}

\def\alphasimplex{{\wideparen{\Delta_m^{\alpha}}}}

\newtheorem{Remark}{Remark}
\newtheorem{Theorem}{Theorem}
\newtheorem{Definition}{Definition}
\newtheorem{Proposition}{Proposition}
\newtheorem{Example}{Example}

\newtheorem{Corollary}{Corollary}

 \newenvironment{BoxedTheorem}
   { \colorlet{shadecolor}{black!10}\begin{shaded}\begin{Theorem}}
   {\end{Theorem}\end{shaded}}

\graphicspath{{Fig/}{./}}

\def\barJ{{\bar J}}
\def\JB{\mathrm{JB}}
\def\JD{\mathrm{JD}}
\def\JS{\mathrm{JS}}

\def\barA{{\bar A}}
\def\barH{{\bar H}}

\def\calP{\mathcal{P}}

\def\ubartheta{{\underline{\theta}}}
\def\bartheta{{\bar\theta}}

\def\barCH{ {\overline{\mathrm{CH}}} }

\def\thetafunc{{\bar\theta}}
\def\lambdafunc{{\bar\lambda}}

\def\ttheta{{\tilde{\theta}}}
\def\tTheta{{\tilde{\Theta}}}

 \def\CN{\mathcal{CN}}
 
\def\calC{\mathcal{C}}
\def\drev{{\star}}
\def\st{{\ :\ }}

\def\eqdef{:=}
\def\LSE{\mathrm{LSE}}
\def\kl{\mathrm{kl}}
\def\dom{\mathrm{dom}}

\def\diag{\mathrm{diag}}

\def\barA{{\bar A}}

\def\barH{{\bar H}}

\def\bartheta{{\bar\theta}}
\def\ubartheta{{\underline{\theta}}}

\def\bbR{\mathbb{R}}
\def\bartheta{\bar\eta}

\def\st{\ :\ }
\def\calP{\mathcal{P}}

\def\calX{\mathcal{X}}
\def\calF{\mathcal{F}}
\def\KL{\mathrm{KL}}

\def\dmu{\mathrm{d}\mu}
\def\Sym{\mathrm{Sym}}

\def\inner#1#2{{\langle #1,#2\rangle}}
\def\Inner#1#2{{\left\langle #1,#2\right\rangle}}
\def\calE{\mathcal{E}}
\def\bartheta{{\bar\theta}}

\def\calA{\mathcal{A}}

\def\inner#1#2{{\langle #1,#2\rangle}}
\def\dmu{\mathrm{d}\mu}
\def\bbR{\mathbb{R}}
\def\inner#1#2{{\langle #1,#2\rangle}}
\def\st{{\ :\ }}
\def\calP{\mathcal{P}}

\def\st{\ :\ }
\def\bbR{\mathbb{R}}

\def\mattwo#1#2#3#4{{\left[\begin{array}{ll}#1 & #2 \cr  #3 & #4 \end{array}\right]}}
\def\st{{\ :\ }}

\def\bbC{\mathbb{C}}

\def\topint{\mathrm{int}}

\def\relint{\mathrm{relint}}
\def\calU{\mathcal{U}}
\def\calA{\mathcal{A}}

\begin{document}
\title{Curved representational Bregman divergences\\ and their applications\footnote{A preliminary version of this work appeared in~\cite{nielsen2025curved}}}

\author{Frank Nielsen\\ \ \\ Sony Computer Science Laboratories Inc.\\ Tokyo, Japan}

\date{}

\maketitle              
\begin{abstract}
By analogy to the terminology of curved exponential families in statistics, we define curved Bregman divergences as Bregman divergences restricted to non-affine parameter subspaces and sub-dimensional Bregman divergences when the restrictions are affine. A common example of curved Bregman divergence is the cosine dissimilarity between normalized vectors: a curved squared Euclidean divergence.
We prove that the barycenter of a finite weighted set of parameters under a curved Bregman divergence amounts to the right Bregman projection onto the non-affine subspace
of the barycenter with respect to the full Bregman divergence, and interpret a generalization of the weighted Bregman centroid of $n$ parameters as a  $n$-fold sub-dimensional Bregman divergence.
We demonstrate the significance of curved Bregman divergences with several examples: (1) symmetrized Bregman divergences, (2) pointwise symmetrized Bregman divergences, (3) additively weighted quadratic Bregman divergences, and (4) the Kullback-Leibler divergence between circular complex normal distributions.
We explain how to reparameterize sub-dimensional Bregman divergences on simplicial sub-dimensional domains.
We then consider monotonic embeddings to define representational curved Bregman divergences and show that the $\alpha$-divergences are representational curved Bregman divergences with respect to $\alpha$-embeddings of the probability simplex into the positive measure cone. 
As an application, we report an efficient method to calculate the intersection of a finite set of  $\alpha$-divergence spheres.
\end{abstract}


\sloppy
%

\section{Introduction and contributions}

Let $V$ be a finite-dimensional Hilbert space  equipped with the inner product denoted by $\inner{\cdot}{\cdot}$.
The inner product induces a norm $\|x\|=\sqrt{\inner{x}{x}}$ and a metric distance $\rho(x,x')=\|x-x'\|$ which endows $V$ with the corresponding metric topology.
The default setting is $V=\bbR^m$ with the Euclidean inner product $\inner{x}{y}=\sum_{i=1}^m x_iy_i$, the  scalar product or dot product.
The Bregman divergence~\cite{Bregman-1967} induced by a strictly convex and differentiable convex function $F:\Theta\subset V\rightarrow\bbR$ is given by
\begin{equation}
B_F(\theta:\theta') \eqdef F(\theta)-F(\theta')-\inner{\theta-\theta'}{\nabla F(\theta')},
\end{equation}
where $\theta\in\Theta$ and $\theta'\in\topint(\Theta)$ (the topological interior of $\Theta$).
Bregman divergences are non-negative dissimilarity measures (i.e., $B_F(\theta:\theta') \geq 0$ with equality if and only if $\theta=\theta'$) which generalize both the squared Mahalanobis distance and the Kullback-Leibler divergence.
Their axiomatization was studied by Csisz\'ar~\cite{Csiszar-1991}.

We consider a class of well-behaved Bregman generators  which are of  Legendre type~\cite{LegendreType-1967,rockafellar-1997}:
A function $F:\Theta\rightarrow\bbR$ is of Legendre-type if  (i) $\Theta$ is topologically open  and  (ii)
 $\lim_{\theta\rightarrow\partial\Theta} \|\nabla F(\theta)\|=\infty$ where $\partial\Theta$ denotes the topological boundary of domain $\Theta$.

When the Bregman generator is of Legendre type~\cite{LegendreType-1967}, the following holds:
(a) there is a one-to-one correspondence between 
 the primal parameter $\theta$ and the corresponding dual parameter $\eta$ defined by $\eta(\theta)=\theta^*=\nabla F(\theta)$, and
(b) the Legendre-Fenchel convex conjugate yields a Legendre-type function: $(H,F^*)$ given by
$F^*(\eta)=\inner{\eta}{\nabla F^{-1}(\eta)}-F(\nabla F^{-1}(\eta))$,
where the open gradient domain is denoted by $H=\{\nabla F(\theta)\st \theta\in\Theta\}$,
 and (c)
  the Legendre-Fenchel transform is an involution, i.e., ${(\Theta,F)^*}^*=(H,F^*)^*=(F,\Theta)$.
	Legendre-type Bregman divergences satisfy the biduality identity~\cite{zhang2004divergence} $B_F(\theta_1:\theta_2)=B_{F^*}^\drev(\eta_1:\eta_2)$ where $D^\drev(x:y)=D(y:x)$ is the reference duality.
 
The paper is organized as follows:
In \S\ref{sec:cbd}, we define curved Bregman divergences (Definition~\ref{def:cbd}).
A characterization of the curved Bregman centroid by Bregman projection is given  in \S\ref{sec:cbc} (Theorem~\ref{thm:proj}).
Several examples illustrating curved Bregman divergences are described: Symmetrized Bregman divergences in \S\ref{sec:symbd}  
 additively weighted quadratic Bregman divergences in \S\ref{sec:addweightquadBD}, the Kullback-Leibler geometry of complex circular normal distributions in \S\ref{sec:CN}.
We further reinterpret and extend the left-sided Bregman centroid in \S\ref{sec:leftsidedBDcentroid} (Proposition~\ref{prop:leftsidedGenBregmanCentroid}).
In \S\ref{sec:crbd}, we further introduce a representation function to define representational curved Bregman divergence (Definition~\ref{def:crbd}).
We prove that $\alpha$-divergences are representational curved Bregman divergences (Proposition~\ref{prop:alphabdrep}).
Finally, as an application, we show how to compute the intersection of a finite set of $\alpha$-divergence spheres in \S\ref{sec:inter} (Proposition~\ref{prop:intersection}).

\begin{figure}
\centering
\includegraphics[width=\textwidth]{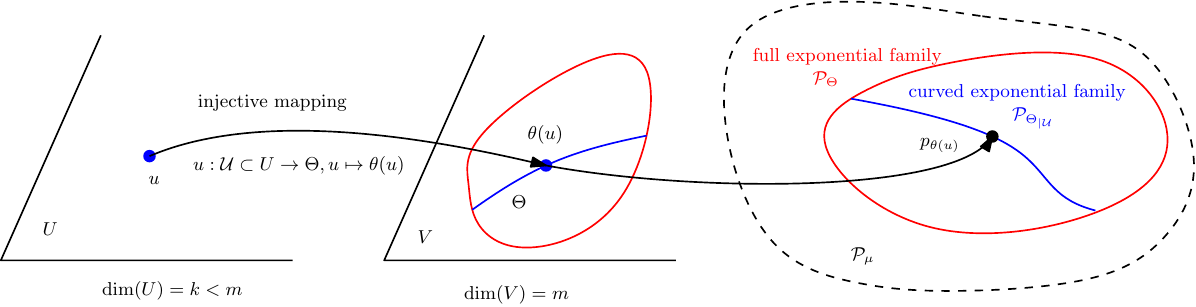}  
\caption{Curved exponential family.}
\label{fig:curvedExpFam}
\end{figure}

\section{Curved Bregman divergences}\label{sec:cbd}

\subsection{Definitions}

Let $(\calX,\calA,\mu)$ be a measure space with $\calX$ the sample space, $\calA$ the $\sigma$-algebra, and $\mu$ a positive measure.
Let us denote by $\calP_\mu$ the set of probability measures dominated by $\mu$ admitting Radon-Nikodym densities.  
The notion of curved exponential families in statistics~\cite{efron1975defining,CEF-Amari-1982} is defined as follows:

\begin{Definition}[Curved exponential family, Figure~\ref{fig:curvedExpFam}]\label{def:curvedEF}
Let $U$ and $V$ be two finite-dimensional vector spaces.
Consider a full exponential family $\calP_\Theta=\{p_\theta=\frac{P_\theta}{\dmu} \st\theta\in\Theta\subset V\}\subset\calP_\mu$ with $\dim(\Theta)=m$.
Then any injective mapping $u:\mathcal{U}\subset U\rightarrow \Theta, u\mapsto \theta(v)$ with $\dim(\calU)<\dim(\Theta)$ defines
a sub-dimensional statistical model $\calP_{\Theta_{|\calU}}=\{ P_{\theta}(u)\st u\in\calU\}$ that is called a curved exponential family.
\end{Definition}

\begin{Example}[\cite{shima2000geometry}]\label{ex:Shima}
Consider a domain $\calU$ of dimension $m'$ and an injective linear mapping $\Sigma(u)$ from $\calU$ to $\Sym^{++}(d,\bbR)$, the set of $d\times d$ symmetric positive-definite matrices.
Then the statistical model 
$\calP'=\left\{p_{\mu,\Sigma(u)} \st \mu\in\bbR^d, u\in\calU \right\}$
is a curved exponential family~\cite{shima2000geometry} with natural parameter $\theta=\Sigma(u)\mu$ for $u\in\calU$.
For example, one can choose $\calU=\bbR_{>0}$ (positive reals) with $\Sigma(u)=u\, I$ ($m'=1$), where $I$ denotes the $d\times d$ identity matrix.
\end{Example}

The notion of curved exponential families allows to consider embedded statistical submanifolds of arbitrary topologies (e.g., Fisher's normal circle model~\cite{efron1978assessing}). 
For example, any parameterized family of distributions in the open standard simplex is a curved exponential family, a sub-model of the full categorical family.
We may also view any $m$-dimensional Riemannian manifold $S$ embedded in some Euclidean space $\mathbb{E}^D$ with $D\leq \frac{m(m+1)}{2}$ (Nash embedding theorem) as a curved exponential family of the isotropic $D$-dimensional Gaussian model (with Fisher-Rao Euclidean geometry).  
Thus although a full statistical model admit generally a global chart, curved statistical models derived thereof can have arbitrary topology.
See also the notion of doubly auto-parallel submanifolds~\cite{ohara2024doubly}.

By analogy, let us define the notion of curved Bregman divergences: 

\begin{Definition}[Curved Bregman divergence]\label{def:curvedBD}\label{def:cbd}
Let $F:\Theta\rightarrow\bbR$ be a strictly convex and differentiable function defined on the domain $\Theta=\dom(F)\subset\bbR^m$ and
 $\calU\subset\Theta$ with $m'=\dim(\calU)<m$ and $\calU=\{\theta(u) \st u\}$.  
Then the curved Bregman divergence induced by $F$ on $\calU$ is defined by: $\forall u_1\in\calU,u_2\in\relint(\calU)$,
\begin{equation}
B_F(u_1:u_2) \eqdef F(\theta(u_1))-F(\theta(u_2))-\inner{\theta(u_1)-\theta(u_2)}{\nabla_u (F(\theta(u_2))}, 
\end{equation}
where $\relint(\calU)$ is the topologically relative interior~\cite{BD-2005} of $\calU$.  
\end{Definition}

\begin{figure}
\centering
\includegraphics[width=\textwidth]{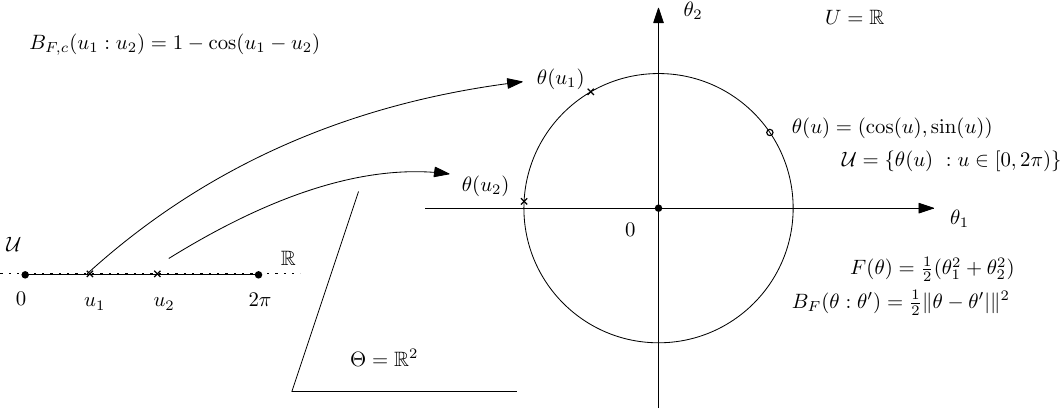}  

\caption{Example of a curved Bregman divergence: The curved circle Euclidean divergence amounts to the cosine dissimilarity.}
\label{fig:curvedCircleEucl}
\end{figure}

\begin{Example}[curved circle Euclidean divergence]\label{ex:cEuclBD}
Consider the Bregman generator $F(\theta)=\frac{1}{2}\, \inner{\theta}{\theta}$ defined in $\Theta=\bbR^2$ yielding the squared Euclidean distance as the corresponding Bregman divergence: 
$B_F(\theta_1:\theta_2)=\frac{1}{2}\, \inner{\theta_2-\theta_1}{\theta_2-\theta_1}=\frac{1}{2}\, \|\theta_2-\theta_1\|_2^2$.
Let $\theta(u)=(\cos(u),\sin(u))$ defining the unit circle $\calU=\{\theta(u)=(\cos(u),\sin(u)) \ :\ u\in\calU=\bbR\}$ (Figure~\ref{fig:curvedCircleEucl}).
Then we get the following curved circle Euclidean divergence:
\begin{eqnarray*}
B_{F,\theta}(u_1:u_2) &=& B_F(\theta(u_1):\theta(u_2)),\\
&=& \frac{1}{2}\, \left( (\cos(u_1)-\cos(u_2))^2 +  (\sin(u_1)-\sin(u_2))^2 \right),\\
&=& 1-\left(\cos(u_1)\cos(u_2)+\sin(u_1)\sin(u_2)\right).
\end{eqnarray*}
Using the trigonometric identity $\cos(u_1)\cos(u_2)+\sin(u_1)\sin(u_2)=\cos(u_1-u_2)$, we end up with the following expression of the curved circle Euclidean divergence:
$$
B_{F,c}(u_1:u_2) =1-\cos(u_1-u_2).
$$
That is, the  curved circle Euclidean divergence is the cosine dissimilarity~\cite{senoussaoui2013efficient} between vectors $u_1$ and $u_2$ on the unit circle $\calU$.
See Appendix~\ref{sec:a:cosdis} for a symbolic code checking that the curved circle Euclidean divergence yields the cosine dissimilarity.

Divergence $B_{F,c}$ is not a Bregman divergence because $F(\theta(u))$ is not a proper Bregman generator (i.e., not strictly convex and differentiable):

\begin{eqnarray*}
F(\theta(u)) &=& \frac{1}{2}\, \inner{\theta(u)}{\theta(u)},\\
&=& \frac{1}{2}\, \inner{(\cos(u),\sin(u))}{(\cos(u),\sin(u))},\\
&=& \frac{1}{2}.
\end{eqnarray*}

This example extends to arbitrary dimension $m$ by defining the $(m-1)$-dimensional sphere 
$\calU=\{ u :\ \|u\|=1 \in\calU=\bbR^m\}$ and $F(\theta)=\frac{1}{2}\, \inner{\theta}{\theta}$.

Notice that the curved circle Euclidean divergence is studied in statistics for the Fisher circle model~\cite{okudo2020bayes} $\{N(\mu,I_2) \ :\ \mu\in\bbR^2\}$, a 1D curved exponential family of the exponential family of bivariate normal distributions with prescribed identity covariance matrix $I_2$.
\end{Example}

\begin{Example}\label{ex:pointwiseBD}
Let $(\calX,\Omega)$ be a measurable space equipped with a positive measure $\mu$ (eg., Lebesgue or counting measure).
Denote by $\calP_\mu$ the set of positive functions $p(x)>0$ such that $\int p(x) \dmu(x)=1$.
Consider a parametric family $\calF=\{p_\theta \st \theta\in\Theta\}$ of $\calP_\mu$ (e.g., Gaussian family or finite mixture family with prescribed components).
For a scalar Bregman divergence $B_f$ with univariate generator $f$, we define the pointwise weighted Bregman divergence~\cite{jones2002general} for a positive weight function $w(x)>0$ as follows:
$$
B_{w,f}^{\mu}(p_l:p_r) = \int w(x)\, B_f(p_l(x),p_r(x)) \, \dmu(x).
$$

Without loss of generality, let us consider $w(x)=1$ for all $x\in\calX$, and denote the corresponding divergence $B_{f}^{\mu}$.
Then we can express the pointwise Bregman divergence $B_{f}^{\mu}$ on $\calF$ as a curved Bregman divergence:
$$
B_{f,p}(\theta_l:\theta_r)=B_{f}^{\mu}(p_{\theta_l}:p_{\theta_r}),
$$
with $p(\theta)=p_\theta$.
\end{Example}

Example of curved Bregman divergences arise from the Kullback-Leibler divergence between parametric curved exponential families.
For example, when considering the spiked covariance model of Gaussians described~\cite{okudo2020bayes,RMTSpikeCovariance-2001} where the family  
$\{N(0,\lambda uu^\top+I) \st \lambda>0, u\in\bbR^m\}$ of centered Gaussians is studied.

\subsection{Sub-dimensional Bregman divergences}

A curved Bregman divergence is a Bregman divergence $B_F$ (with $\mathrm{dom}(F)=\Theta$) defined on a non-linearly constrained parameter sub-dimensional domain $\Omega\subset\Theta$: It is thus never a Bregman divergence because the domain $\Omega$ of the curved Bregman divergence is not anymore convex.
When the domain constraints are linear, we call the restricted Bregman divergence a sub-dimensional Bregman divergence: The sub-dimensional Bregman divergence as a different Bregman generator than the full-dimensional Bregman divergence.
Indeed, we can view the sub-dimensional Bregman degenerator by considering the graph of $F$ restricted to $\Omega$: This subgraph can be interpreted as the graph on domain $\Omega$ of another function $F_\Omega(\omega)$.

 Notice that we reparameterize a multivariate Bregman divergence as a family of scalar Bregman divergences:

\begin{Proposition}\label{prop:mv2uvBD}
A multivariate Bregman divergence $B_F(\theta:\theta')$  
can be written equivalently as a scalar Bregman divergence $B_{F_{\theta,\theta'}}(0:1)$ so that 
$$
\forall \theta\in\Theta,\theta'\in\topint(\theta),\quad B_F(\theta:\theta') = B_{F_{\theta,\theta'}}(0:1),
$$ 
where $F_{\theta,\theta'}(u) := F(\theta+u(\theta'-\theta))$
is a univariate Bregman generator.
\end{Proposition}

\begin{proof}
Let us prove that a $m$-variate Bregman divergence $B_F(\theta:\theta')$ may be considered as a family of $1$-dimensional Bregman divergences defined on 1D line segment domains $\calC=[\theta,\theta']$  
with Bregman generators
${F_{\cal_C}}(\lambda)=   F_{\theta,\theta'}(\lambda)= F((1-\lambda)\theta+\lambda\theta')$ for the interval domain $\Delta_1=[0,1]$:  
$B_F(\theta:\theta')=B_{F_{\theta,\theta'}}(0:1)$.

Indeed, consider the directional derivative:
\begin{eqnarray*}
\nabla_{\theta'-\theta}\, F_{\theta,\theta'}(u) &\eqdef& \lim_{\epsilon\rightarrow 0} 
\frac{F(\theta+(\epsilon+u)(\theta'-\theta))-F(\theta+u(\theta'-\theta))}{\epsilon},\\
&=& \inner{\theta'-\theta}{\nabla F(\theta+u(\theta'-\theta))}.
\end{eqnarray*}
This result can be recovered using first-order Taylor approximations:
$$
F(\theta+\lambda(\theta'-\theta))\approx_{\lambda\rightarrow 0} F(\theta)+\lambda \inner{\theta'-\theta}{\nabla F(\theta)}.
$$

Since $F_{\theta,\theta'}(0)=F(\theta)$, $F_{\theta,\theta'}(1)=F(\theta')$, and 
$F_{\theta,\theta'}'(u)=\nabla_{\theta'-\theta} F_{\theta,\theta'}(u)$, it follows that we have
\begin{eqnarray*}
B_{F_{\theta,\theta'}}(0:1)&\eqdef&F_{\theta,\theta'}(0)-F_{\theta,\theta'}(1)-(0-1)\nabla_{\theta'-\theta} F_{\theta,\theta'}(1),\\
&=& F(\theta)-F(\theta')+\inner{\theta'-\theta)}{\nabla F(\theta')} = B_F(\theta:\theta').
\end{eqnarray*}
\end{proof}

Thus a multivariate Bregman divergence $B_F(\theta:\theta')$ for $\theta'\not=\theta$ may be represented by a family of 
univariate Bregman divergences: $\{B_{F_{\theta,\theta'}} \st (\theta,\theta')\in \Theta\times\topint(\Theta), \theta'\not=\theta\}$.

\begin{figure}
\centering
\includegraphics[width=0.95\columnwidth]{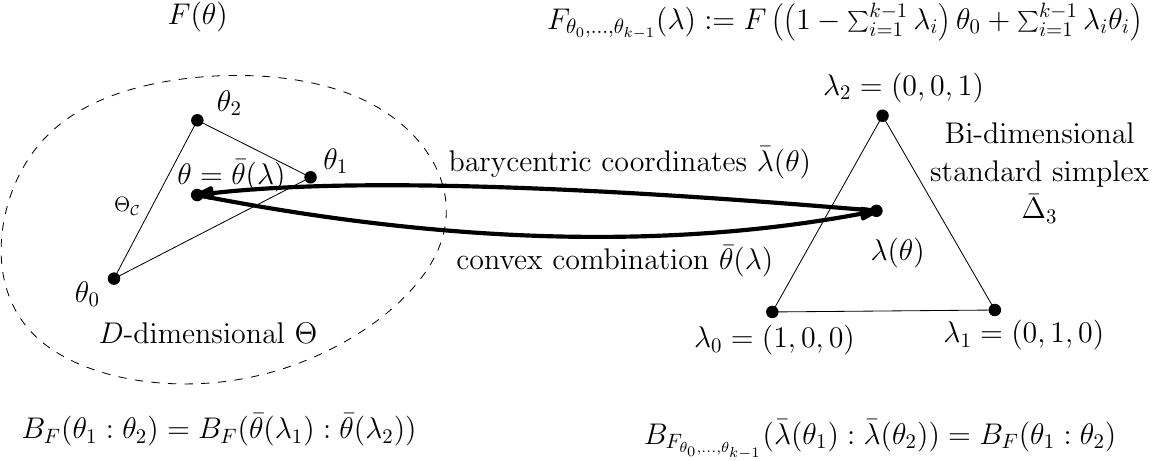}
\caption{Illustration of the reparameterization and sub-dimensional Bregman divergences.\label{fig:subBD}}
\end{figure}

Consider the more general reparameterization case:
Let $k$ parameters $\theta_0,\ldots,\theta_{k-1}$ be in convex position in $\Theta$.
That is, the parameters $\theta_i$'s are the distinct extreme points of the closed convex hull domain $\calC=\barCH(\theta_0,\ldots,\theta_{k-1})$ with
$$
\barCH(\theta_0,\ldots,\theta_{k-1}) =
\left\{  \sum_{i=0}^{k-1} \lambda_i\theta_i \st \lambda_i\geq 0, \sum_{i=0}^{k-1} \lambda_i=1 \right\}.
$$

We may reparameterize the function $F$ restricted to $\calC$ (i.e., ${F_\calC}:\calC\subset\bbR^m\rightarrow\bbR$) 
according to $k-1$ free parameters $\lambda_1,\ldots,\lambda_{k-1}$ since $\lambda_{0}=1-\sum_{i=1}^{k-1} \lambda_i$ (Figure~\ref{fig:subBD}).

\begin{Proposition}[Reparameterized/sub-dimensional Bregman divergence]\label{thm:subBD}
Let $B_F$ be a Bregman divergence for the Legendre-type Bregman generator $F:\Theta\subset\bbR^m\rightarrow\bbR$ with $m=\dim(\Theta)$.
Consider $k$ parameters $\theta_0,\ldots,\theta_{k-1}$ of $\topint(\Theta)$ in convex position, 
and let $\calC=\barCH(\theta_0,\ldots,\theta_{k-1})$ denote their closed convex hull.
For $\lambda\in\bar\Delta_{k-1}$ the closed $(k-1)$-dimensional standard simplex, 
let $\thetafunc(\lambda)\eqdef\sum_{i=1}^{k-1} \lambda_i\theta_i + \left(1-\sum_{i=1}^{k-1} \lambda_i\right)\,\theta_{0}$ (i.e., interpolation using a weighted simplex vertex mean).
We have $\Theta_\calC=\{\bar\theta(\lambda)\ :\ \lambda\in \bar\Delta_{k-1}\}$. 

Conversely, for $\theta\in\Theta_\calC$, let $\lambdafunc(\theta)=(\lambda_0,\ldots,\lambda_{k-1})$ be the barycentric coordinates of $\theta$ 
 such that $\thetafunc(\lambdafunc(\theta))=\theta$.
Denote by $F_{\theta_0,\ldots,\theta_{k-1}}(\lambdafunc(\theta)) = F(\bar\theta(\lambda))$ the $(k-1)$-dimensional Bregman generator.
Then we have the Bregman generator $F_{\theta_0,\ldots,\theta_{k-1}}: \bar\Delta_{k-1} \rightarrow\bbR$, and the following identities hold: 
\begin{equation}
\forall \lambda,\lambda'\in\bar\Delta_{k-1}, \quad 
B_{F_{\theta_0,\ldots,\theta_{k-1}}}(\lambda:\lambda')=B_F(\bar\theta(\lambda):\bar\theta(\lambda')),
\end{equation}
and reciprocally, it holds that
\begin{equation}
\forall \theta,\theta'\in\Theta_\calC, \quad B_F(\theta:\theta')=B_{F_{\theta_0,\ldots,\theta_{k-1}}}(\lambdafunc(\theta):\lambdafunc(\theta')).
\end{equation}
\end{Proposition}

\begin{proof}
Figure~\ref{fig:subBD} illustrates the reparameterization of a multivariate Bregman divergence restricted to a closed simplex domain of $\Theta$.
We term the Bregman generator $F_{\theta_0,\ldots,\theta_{k-1}}$ a simplex Bregman generator.
When the dimension of the simplex is strictly less than then dimension of the domain $\Theta$, the reparameterized Bregman divergence is called a sub-dimensional Bregman divergence.

To prove the identities, consider the partial derivatives of the simplex Bregman generator with respect to $\lambda$:
$$
F_{\theta_0, \ldots, \theta_{k-1}}(\lambda_1,\ldots,\lambda_{k-1}) 
=  F\left(\theta_0 + \sum_{i=1}^{k-1}\lambda_i (\theta_i-\theta_0) \right).
$$
These partial derivatives $\frac{\partial}{\partial\lambda_i} F_{\theta_0, \ldots, \theta_{k-1}}$ are equivalent to the following corresponding directional derivatives:
$$
\frac{\partial}{\partial\lambda_i} F_{\theta_0, \ldots, \theta_{k-1}}(\lambda)
=
(\theta_i-\theta_0)^\top \nabla_{\theta_i-\theta_0} F\left(\theta_0 + \sum_{i=1}^{k-1}\lambda_i (\theta_i-\theta_0) \right).
$$

Thus we get the gradient of the simplex Bregman generator as follows:
$$
\nabla_\lambda {F_{\theta_0, \ldots, \theta_{k-1}}}(\lambda)=\left[\begin{array}{c}
(\theta_1-\theta_0)^\top \nabla F\left(\theta_0 + \sum_{i=1}^{k-1}\lambda_i (\theta_i-\theta_0) \right)\\
\vdots\\
(\theta_{k-1}-\theta_0)^\top \nabla F\left(\theta_0 + \sum_{i=1}^{k-1}\lambda_i (\theta_i-\theta_0) \right)
\end{array}\right].
$$

Since we have
$$
\theta-\theta'=\thetafunc(\lambda)-\thetafunc(\lambda')=\sum_{i=1}^{k-1} (\lambda_i-\lambda_i')(\theta_i-\theta_0),
$$
it follows that we get:
\begin{eqnarray*}
\inner{\theta-\theta'}{\nabla F(\theta')} &=&  \Inner{\sum_{i=1}^{k-1} (\lambda_i-\lambda_i')(\theta_i-\theta_0)}{\nabla F(\theta')},\\
&=&  \sum_{i=1}^{k-1} (\lambda_i-\lambda_i')(\theta_i-\theta_0) \frac{\partial}{\partial\theta_i} F(\theta'),\\
&=&  \Inner{\lambda-\lambda'}{\nabla_\lambda {F_{\theta_0, \ldots, \theta_{k-1}}}(\lambda')}.
\end{eqnarray*}

Hence,  we get with $\lambda=\bar\lambda(\theta)$ and $\lambda'=\bar\lambda(\theta')$, and $F(\theta)=F_{\theta_0,\ldots,\theta_{k-1}}(\lambda)$ and 
 $F(\theta')=F_{\theta_0,\ldots,\theta_{k-1}}(\lambda')$:

\begin{eqnarray*}
B_{F}(\theta:\theta') & \eqdef& F(\theta)-F(\theta')-\inner{\theta-\theta'}{\nabla F(\theta')},\\
&=&  F_{\theta_0,\ldots,\theta_{k-1}}(\lambda)-F_{\theta_0,\ldots,\theta_{k-1}}(\lambda')-\Inner{\sum_{i=1}^{k-1} (\lambda_i-\lambda_i')(\theta_i-\theta_0)}{\nabla F(\theta')},\\
&=& F_{\theta_0,\ldots,\theta_{k-1}}(\lambda)-F_{\theta_0,\ldots,\theta_{k-1}}(\lambda')-\Inner{\lambda-\lambda'}{\nabla_\lambda {F_{\theta_0, \ldots, \theta_{k-1}}}(\lambda')},\\
&=& B_{F_{\theta_0,\ldots,\theta_{k-1}}}(\lambda:\lambda').
\end{eqnarray*}
The Bregman divergence $B_{F}$ on the left hand side of the equality $B_{F}(\theta:\theta')=B_{F_{\theta_0,\ldots,\theta_{k-1}}}(\lambda:\lambda')$ is $m$-variate while the Bregman divergence $B_{F_{\theta_0,\ldots,\theta_{k-1}}}$ on the right hand side is $(k-1)$-variate (for $k\leq m+1$).
\end{proof}

\begin{Example}[KLD as a sub-dimensional Bregman divergence]\label{ex:kldsubbd}
The Kullback-Leibler divergence on the $(m-1)$-dimensional simplex is a sub-dimensional Bregman divergence in disguise of the KLD on the positive cone of $\bbR^m_{>0}$. 
Indeed, consider the full dimensional Bregman divergence  defined on the positive orthant $\tTheta=\bbR_{>0}^m$ for the  Bregman generator 
$F_{\KL^+}(\theta)=\sum_{i=1}^m \theta_i\log\theta_i$.
This Bregman generator $F_{\KL^+}$ is said separable or decomposable as it can be expressed as a sum of scalar Bregman generators:
Namely, we have
$F_{\KL^+}(\theta)=\sum_{i=1}^m f_{\kl^+}(\theta_i)$ where $f_{\kl^+}: \bbR_{>0} \rightarrow \bbR$ maps $\theta$ to $f_{\kl^+}(\theta)=\theta\log\theta$ (Shannon negentropy scalar function). We denote by $S^+(\theta)=-F_{\KL^+}(\theta)$ Shannon entropy relaxed to positive measures.
The corresponding Bregman divergence is the Kullback-Leibler divergence extended to the positive orthant cone~\cite{Csiszar-1991} (i.e., the set of positive arrays~\cite{IG-2016} which includes the probability simplex):
$$
B_{F_\KL^+}(\ttheta:\ttheta') = \sum_{i=1}^{m} B_{F_{\kl^+}}(\ttheta_i:\ttheta_i')
  = \sum_{i=1}^{m}  \ttheta_i\log\frac{\ttheta_i}{\ttheta'_i}+\ttheta_i'-\ttheta_i
:= D_{\KL^+}[\ttheta:\ttheta'].
$$

Let $e_i=(0,\ldots,0,1,0,\ldots,0)$ be the $i$th one hot vector of the natural basis $\{e_1,\ldots,e_m\}$ of $\bbR^m$. 
Let $\lambdafunc(\theta)$ be the barycentric coordinates of $\theta$ with respect to $e_1,\ldots,e_{m-1}$, and 
denote by $\Lambda_{m-1}=\bar\Delta_{m-1}$ the set of barycentric parameters.
Then  $\thetafunc(\lambda)\eqdef\sum_{i=1}^{m-1} \lambda_i e_i+ (1-\sum_{i=1}^{m-1} \lambda_i)e_m$ and 
$$
F_{\KL_{\Delta_m}}(\lambda)={F_\KL^+}(\thetafunc(\lambda)).
$$
Notice that the Bregman generator $F_{\KL_{\Delta_m}}$ of the sub-dimensional Bregman divergence
$B_{F_\KL^+}$  is not separable because of the normalization constraint: $\sum_{i=1}^{m}\lambda_i=1$.
Figure~\ref{fig:subKLD} illustrates the extended KLD (separable Bregman divergence) defined on the positive orthant and its restriction on the $(k-1)$-dimensional simplex $\Delta_k$ yielding a sub-dimensional non-separable Bregman divergence.

Here, the simplex $\Delta_2$ is considered as a mixture family with Bregman generator 
$F(\lambda)=\sum_{i=1}\lambda_i\log \lambda_i=-S(\lambda_1,\lambda_2,1-(\lambda_1+\lambda_2))$.
Alternatively, the simplex $\Delta_2$ can also be considered as an exponential family with  Bregman generator 
$F(\alpha)=\LSE_+(\alpha_1,\alpha_2)$  (dual generator) where $\LSE_+(\alpha_1,\alpha_2)=\LSE(0,\alpha_1,\alpha_2)=\log(1+e^{\alpha_1}+e^{\alpha_2})$ is a strictly convex function (the log-sum-exp function LSE is only a convex function~\cite{boyd2004convex}).
The parameters $\alpha_i=\log\frac{\lambda_i}{\lambda_3}$ are called the natural parameters of the exponential family while the parameters $\lambda_i$ are 
called the dual moment parameters (for $i\in\{1,2\}$).
\end{Example}

\begin{figure}
\centering
\includegraphics[width=0.75\columnwidth]{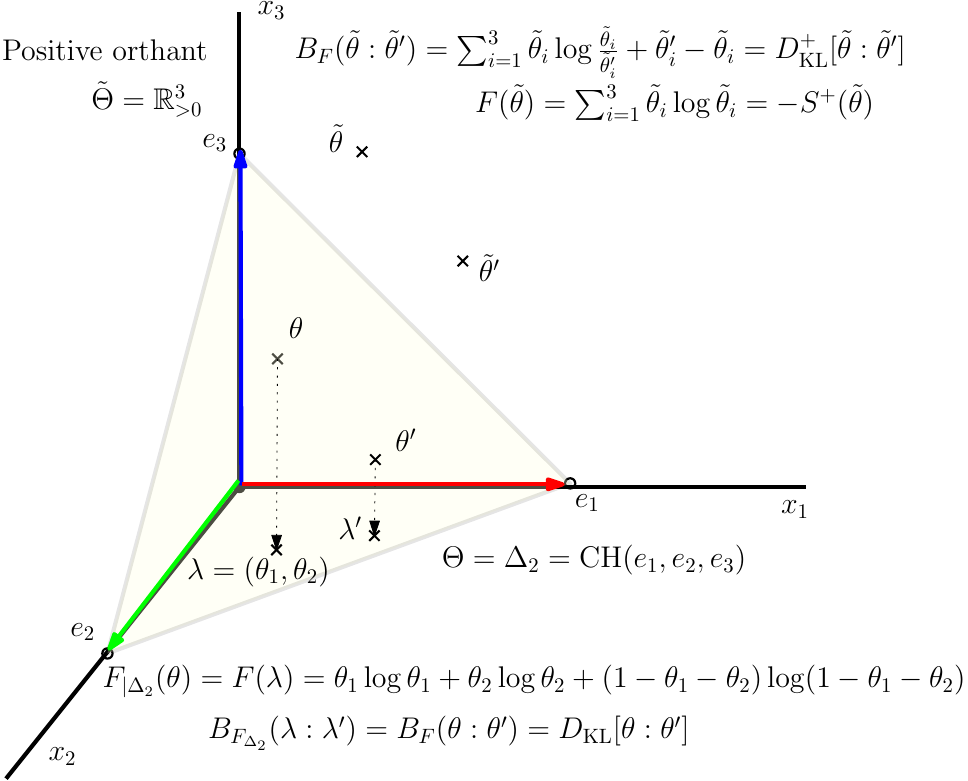}

\caption{The discrete Kullback-Leibler divergence between two probability mass functions of the 2D simplex $\Delta_2$ sitting in $\bbR_{>0}^3$ is an example of sub-dimensional Bregman divergence of the extended KLD defined between positive arrays of $\tTheta=\bbR_{>0}^3$.\label{fig:subKLD}}
\end{figure}

\subsection{Curved Bregman centroids}\label{sec:cbc}

\begin{figure}
\centering
\includegraphics[width=\textwidth]{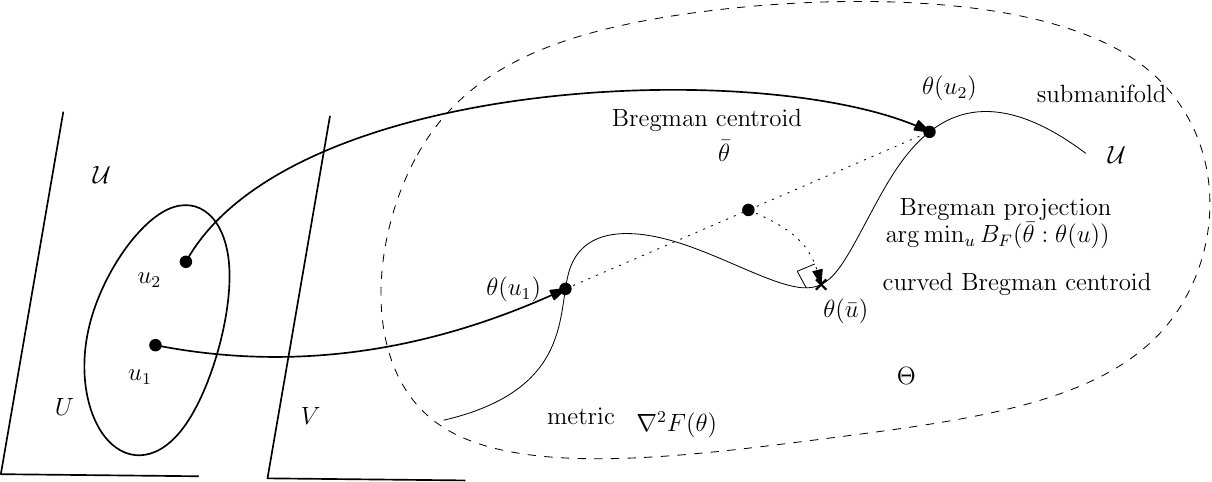}  
\caption{The curved Bregman centroid $\theta(\bar{u})$ amounts to the right Bregman projection of the unconstrained Bregman centroid $\bar\theta$ onto the subspace $\calU$.}
\label{fig:curvedBregmanCentroid}
\end{figure}

We show that the right curved Bregman centroid is characterized by the Bregman projection~\cite{nielsen2018information} of the Bregman centroid onto $\calU$:

\begin{BoxedTheorem}\label{thm:proj}
Let $\theta_i=\theta(u_i)$'s be $n$ weighted parameters of $\calU$ with weight vector $w\in\Delta_{n-1}$ (the $(n-1)$-dimensional standard simplex).
Then the barycenter in $\calU$ with respect to the curved Bregman divergence amounts to the Bregman projection of the center of mass $\bar\theta=\sum_i w_i\theta_i$
 (right Bregman barycenter) onto $\calU$:
\begin{equation}
\arg \min_{u\in\calU} \sum_{i=1}^n w_i\, B_F(\theta_i:\theta(u)) = \arg\min_{u\in\calU}\, B_F(\bar\theta:\theta(u)).
\end{equation}
\end{BoxedTheorem}

\begin{proof}
\begin{eqnarray*}
\min_{u\in\calU} \sum_{i=1}^n w_i\, B_F(\theta_i:\theta(u))  &=& \sum_{i=1}^n w_i \left(F(\theta_i)-F(\theta(u)) -\inner{\theta_i-\theta(u)}{\nabla F(\theta(u))} \right),\\
&\equiv & -F(\theta(u)) -\inner{\bar\theta-\theta(u)}{\nabla F(\theta(u))},\\
&\equiv & F(\bar\theta)-F(\theta(u)) -\inner{\bar\theta-\theta(u)}{\nabla F(\theta(u))},\\
&=& B_F(\bar\theta:\theta(u)).
\end{eqnarray*}
\end{proof}


Figure~\ref{fig:curvedBregmanCentroid} illustrates the right Bregman projection with orthogonality defined with respect to the Hessian metric~\cite{shima2007geometry,shima2013geometry} $\nabla^2 F(\theta)$ to find the curved Bregman centroid. Theorem~\ref{thm:proj} generalizes the result of Amari and Nagaoka~\cite{amari2000methods} of the maximum likelihood estimators on curved exponential families. 

\begin{Example}
Let us consider the curved circle Euclidean divergence of Example~\ref{ex:cEuclBD}:
Consider a set $u_1,\ldots, u_n$ of reals which yield the corresponding parameters $\theta(u_1)=(\cos(u_1),\sin(u_1))$, $\ldots$, $\theta(u_n)=(\cos(u_n),\sin(u_n))$.
The right Bregman centroid with respect to $B_F$ is the center of mass:
$$
\bar\theta=\frac{1}{n}\sum_{i=1}^n (\cos(u_i),\sin(u_i))=:(\bar{c},\bar{s}),
$$
where $\bar{c}=\sum_{i=1}^n \frac{1}{n}\cos(u_i)$ and  $\bar{s}=\sum_{i=1}^n \frac{1}{n}\cos(u_i)$.

The right curved centroid with respect to $u_1,\ldots,u_n$ is defined by
\begin{eqnarray*}
\arg\min_u && \sum_{i=1}^n \frac{1}{n}\, B_{F,c}(u_i:u),\\
 && \sum_{i=1}^n \frac{1}{n}\, \sum_{i=1}^n (1-\cos(u)\cos(u_i)-\sin(u)\sin(u_i)).
\end{eqnarray*}

The minimization problem is equivalent to the following maximization problem:
$$
\max_u  \bar{c}\cos(u)+\bar{s}\sin(u).
$$

The extremal point is found by setting the derivative to zero:
$$
\bar{c}\sin(u)-\bar{s}\cos(u)=0.
$$

Thus we get the solution $u=\arctan \frac{\bar{s}}{\bar{c}}$,

Let us check that
$$
\arg\min_u \sum_{i=1}^n \frac{1}{n}\, B_F(\theta_i:\theta(u))=\arg\min_u  B_F(\bar\theta:\theta(u)),
$$
where $\bar\theta=(\bar{c},\bar{s})$.

Notice that the curved BD centroid theorem recovers the fact that the cosine (dis)similarity centroid is the center of mass reprojected onto the unit sphere.
\end{Example}

\begin{Remark}
Although the full Bregman centroid $\bar\theta$ is unique, the curved Bregman centroid(s) may not be unique when the projection is not unique: e.g., consider $\theta_i$ on a unit circle with $\bar\theta$ coinciding with the circle center (see also Fisher's normal circle model~\cite{efron1978assessing}). Then the projection can be any point of the circle.
That is, let $B_F(\theta_1:\theta_2)=\frac{1}{2}\inner{\theta_2-\theta_1}{\theta_2-\theta_1}=\frac{1}{2}\|\theta_1-\theta_2\|^2$ be the Bregman divergence obtained for $F(\theta)=\frac{1}{2}\inner{\theta}{\theta}$. Consider $U=\bbR$ and $\calU=[0,2\pi)\subset U$ with $\theta(u)=(\cos\theta,\sin\theta)\in\bbR^2=:V$.
We have $B_F(u_1:u_2)=\frac{1}{2}((\cos u_2-\cos u_1)^2+(\sin u_2-\sin u_1)^2)$. Let $u_1=0$ and $u_2=\pi$ with $\theta_1=(1,0)$ and $\theta_2=(-1,0)$.
We have $\bar\theta=(0,0)$ and $\arg\min_{u} B_F((0,0):\theta(u))=\arg\min_{u} \frac{1}{2}\|\cos^2\theta+\sin^2\theta\|^2=\arg\min_{u} \frac{1}{2}=\calU$, i.e., the curved Bregman centroid is the full circle $\{\theta(u)\st u\in\calU\}$.
\end{Remark}

\begin{Corollary}[\cite{amari2000methods}]
Let $x_1,\ldots, x_n$ be $n$ i.i.d. observations of a curved exponential family $\{p_{\theta(u)}(x)\st u\in\Theta_\calU\}$ with sufficient statistic vector $t(x)$  and cumulant function $F(\theta)$.
Then the maximum likelihood estimator  amounts to a curved Bregman centroid on $t(x_1),\ldots, t(x_n)$ with respect to the convex conjugate $F^*(\eta)$.
That is, the right dual  Bregman projection of the observed point $\bar t=\frac{1}{n} \sum_{i=1}^n t(x_i)$ onto the moment submanifold $\{\eta(u)=\nabla F(\theta(u))\}$:
$$
\eta(\hat u)=\min_{u\in\calU} B_{F^*}(\bar t:\eta(u)).
$$
\end{Corollary}

\begin{proof}
For regular full exponential families (i.e., $\theta(u)$ the identity function), using the duality between regular exponential families and regular Bregman divergences~\cite{BD-2005}, we have~\cite{nielsen2012k}:
$$
\max_{\theta\in\Theta} \prod_{i=1}^n p_\theta(x_i) \equiv \min_{\eta} \sum_{i=1}^n B_{F^*}(t(x_i):\eta).
$$

Let $\eta(u)=\nabla F(\theta(u))$. It follows that we have
$$
\max_{u} \prod_{i=1}^n p_{\theta(u)}(x_i) \equiv \min_{u} B_{F^*}(\bar t:\eta(u)),
$$
where $\bar t=\frac{1}{n} \sum_{i=1}^n t(x_i)$ is the observed point~\cite{amari2000methods}.
\end{proof}

Let $I_F(\theta_1,\ldots,\theta_n;\theta) = \sum_{i=1}^n \frac{1}{n} B_F(\theta_i:\theta)$
The Bregman information~\cite{BD-2005} is defined as $I_F(\theta_1,\ldots,\theta_n)=I_F(\theta_1,\ldots,\theta_n;\bar\theta)$ where $\bar\theta= \sum_{i=1}^n \frac{1}{n}\theta_i$ is the center of mass.
We have the following bias-variance decomposition~\cite{BD-2005,pfau2025generalized,heskes2025bias}:
$$
I_F(\theta_1,\ldots,\theta_n;\theta) - I_F(\theta_1,\ldots,\theta_n;\bar\theta) =  B_F(\bar\theta:\theta)\geq 0.
$$
Hence, the Bregman centroid $\bar\theta$ minimizes the average Bregman divergence and its minimal value corresponds to the Bregman information.
The Bregman bias-variance decomposition is thus:
$$
I_F(\theta_1,\ldots,\theta_n;\theta) = I_F(\theta_1,\ldots,\theta_n;\bar\theta) + B_F(\bar\theta:\theta).
$$

Let $u^*$ be one minimizer of $\min_{u\in\calU}\, B_F(\bar\theta:\theta(u))$.
Then we can apply the Bregman bias-variance decomposition to get
$$
I_F(\theta_1,\ldots,\theta_n;\theta(u^*)) = B_F(\bar\theta:\theta(u^*))+ I_F(\theta_1,\ldots,\theta_n).
$$

The left-hand-side may be called the curved Bregman information. It decomposes as the Bregman information term (induced by the Bregman centroid) 
plus the term $B_F(\bar\theta:\theta(u^*))$ which indicates the Bregman loss when considering a constraint $\calU$ instead of the full parameter space.

Interestingly, the Pythagorean theorem also holds for pointwise Bregman divergences~\cite{jones2002general}.

Now, let us report now  some examples of curved Bregman divergences.

\subsection{Symmetrized Bregman divergences as curved Bregman divergences}\label{sec:symbd}

Bregman divergences are proven asymmetric except when the generator is quadratic (yielding a generalized Euclidean divergence).
See~\cite{nielsen2007bregman} (Lemma~2) and~\cite{jones2002general} (Theorem~2).

\begin{Example}
When $F(\theta)=\frac{1}{2}\theta^\top Q\theta=F_Q(\theta)$ for a positive-definite matrix $Q\succ 0$, we have the convex conjugate $F^*(\eta)=\frac{1}{2}\eta^\top Q^{-1}\eta$ (with $Q^{-1}\succ 0$). We have 
$\eta_i=Q^{-1}\theta_i$ and $\eta_i=Q\theta_i$.
It follows that $\bar\theta=\sum_{i=1}^n w_i\theta_i=Q^{-1}\bar\eta$
 and
$\bar\eta=\sum_{i=1}^n w_i\eta_i=Q\bar\theta$.
Thus we check that the information radii coincide when dealing with squared Mahalanobis Bregman divergences:
\begin{eqnarray}
I_{F_Q}(\theta_1,\ldots,\theta_n) &=& \sum_{i=1}^n w_i \frac{1}{2}\theta_i^\top Q\theta_i - \frac{1}{2}\bar\theta^\top Q\bar\theta,\\
&=&  \sum_{i=1}^n w_i \frac{1}{2}(Q^{-1}\eta_i)^\top Q (Q^{-1}\eta_i) - \frac{1}{2} (Q^{-1}\bar\eta)^\top Q (Q^{-1}\bar\eta),\\
&=& \sum_{i=1}^n w_i \eta_i^\top Q^{-1}\eta_i - \frac{1}{2} \bar\eta Q^{-1}\bar\eta,\\
&=& I_{F_{Q^{-1}}}(\eta_1,\ldots,\eta_n).
\end{eqnarray}
\end{Example}

Consider the symmetrized Bregman divergence for Legendre-type generators $F$:
\begin{eqnarray}
S_F(\theta_1,\theta_2) &:=& B_F(\theta_1:\theta_2) + B_F(\theta_2:\theta_1),\label{eq:jbd}\\
 &=&  B_F(\theta_1:\theta_2) + B_{F^*}(\nabla F(\theta_1):\nabla F(\theta_2)),\\
 &=& B_{F_\xi}(\xi(\theta_1):\xi(\theta_2)),
\end{eqnarray} 
where $\theta_i\in\Theta\subset\bbR^m$, $\xi(\theta):=(\theta,\nabla F(\theta))\in\bbR^{2m}$
and 
$$
F_\xi(\theta):=F(\theta)+F^*(\nabla F(\theta)),
$$ 
where $F^*(\eta)=\inner{\theta(\eta)}{\eta}-F(\theta(\eta))$ denotes  the convex conjugate of $F(\theta)$ with $\theta(\eta)=(\nabla F)^{-1}(\eta)$.

The symmetrized Bregman divergence $S_F$ (also called Jeffreys-Bregman divergence~\cite{nielsen2019jensen}) is a curved Bregman divergence $B_{F_\xi}$ with respect to the generator $F_\xi$ and
 the (usually non-linear) subspace $\calU=\{u(\theta)=(\theta,\nabla F(\theta)) \st \theta\in\Theta\}$.
Thus the symmetrized Bregman centroid~\cite{nielsen2009sided} of a weighted parameter set $\{\theta_i\}$ is the projection of the centroid $\bar\xi=(\sum w_i \theta_i,\sum w_i \nabla F(\theta_i))$ onto $\calU=\{(\theta,\nabla F(\theta)) \st \theta\in\Theta\}$:
\begin{eqnarray*}
\arg \min_{\theta\in\Theta} \sum_i w_i\, S_F(\theta_i,\theta)
&=&
\arg \min_{\theta\in\Theta} \sum_i w_i\, B_{F_\xi}(\xi_i:\xi(\theta)),\\
 &=& \arg\min_{u\in\calU} B_{F_\xi}(\bar\xi:u),
\end{eqnarray*}
where $\xi_i=(\theta_i,\eta_i=\nabla F(\theta_i))$ and $\bar\xi= \sum_i w_i\xi_i=(\bar\theta,\bar\eta)$.

Thus the symmetrized Bregman centroid~\cite{nielsen2009sided} is a curved Bregman centroid, and we have:
$$
\min_{\theta\in\Theta} \sum_{i=1}^n w_i\, B_{F_\xi}(\xi_i:\xi(\theta))  \equiv \min_{\theta\in\Theta} B_{F_\xi}(\bar\xi:\xi(\theta)),
$$
where $\bar\xi=(\bar\theta= \sum_{i=1}^n w_i\theta_i,\bar\eta= \sum_{i=1}^n w_i\eta_i)$.
Since $B_{F_\xi}(\bar\xi:\xi(\theta))=B_F(\bar\theta:\theta)+B_{F^*}(\bar\eta:\eta) = B_F(\bar\theta:\theta)+B_F(\theta:\ubartheta)$ with
$\ubartheta=\nabla F^*(\bar\eta)$, it follows that we have
$$
\min_{\theta\in\Theta} \sum_{i=1}^n w_i\, B_{F_\xi}(\xi_i:\xi(\theta))  \equiv \min_{\theta\in\Theta} B_F(\bar\theta:\theta)+B_F(\theta:\ubartheta).
$$

Thus right Bregman projection perspective of symmetrized Bregman centroids let us recover Lemma 4.1 of~\cite{nielsen2009sided}:

\begin{Corollary}[Lemma 4.1 of~\cite{nielsen2009sided}]\label{sec:lemmasymBD}
The symmetrized Bregman barycenter of  $n$ weighted parameters $\theta_1,\ldots, \theta_n$ weighted by $w=(w_1,\ldots,w_n)\in\Delta_{n}$ (standard simplex) minimizes 
$\min_{\theta\in\Theta} B_F(\bar\theta:\theta)+B_F(\theta:\ubartheta)$ where $\bar\theta=\sum_i w_i\theta_i$ (right Bregman centroid) and $\ubartheta=\sum_i w_i \nabla F(\theta_i)$ (left Bregman centroid).
\end{Corollary}

The right-hand-side amounts to minimize equivalently
$$
\min_{\theta\in\Theta} \inner{\theta-\bartheta}{\nabla F(\theta)}-\inner{\theta}{\nabla F(\ubartheta)}.
$$

\begin{Remark}
There are many ways to symmetrize an asymmetric divergence $D(\theta_l:\theta_r)$ beyond taking the arithmetic mean of the forward and reverse divergences of Eq.~\ref{eq:jbd}.
For example, we may symmetrize as follows:
\begin{eqnarray*}
D_J(\theta_l,\theta_r) &=& D(\theta_l:\theta_r)+D(\theta_r:\theta_l) = \JD(\theta_r,\theta_l),\\
D_\JS(\theta_l,\theta_r) &=&  D\left(\theta_l:\frac{\theta_l+\theta_r}{2}\right) + D\left(\theta_r:\frac{\theta_l+\theta_r}{2}\right) = D_\JS(\theta_r,\theta_l),\\
D_C(\theta_l,\theta_r) &=& \min_\theta \max\{ D(\theta_l:\theta_r), D(\theta_r:\theta_l)\}.
\end{eqnarray*}
When $D$ is the Kullback-Leibler divergence, $D_J$ is called the Jeffreys divergence, $D_\JS$ the Jensen-Shannon divergence, and $D_C$ the Chernoff information.
Thus when $D=B_F$, we get the Jeffreys-Bregman divergence $S_F$, the Jensen-Bregman divergence $\JB_F$, and the Chernoff-Bregman divergence $C_F$ as follows:
\begin{eqnarray*}
S_F(\theta_l,\theta_r) &=& B_F(\theta_l:\theta_r)+B_F(\theta_r:\theta_l) = S_F(\theta_r,\theta_l),\\
\JB_F(\theta_l,\theta_r) &=&  B_F\left(\theta_l:\frac{\theta_l+\theta_r}{2}\right) + B_F\left(\theta_r:\frac{\theta_l+\theta_r}{2}\right) = \JB_F(\theta_r,\theta_l),\\
C_F(\theta_l,\theta_r) &=& \min_\theta \max\{ B_F(\theta_l:\theta_r), B_F(\theta_r:\theta_l)\}.
\end{eqnarray*}
Notice that $\JB_F(\theta_l,\theta_r)=\frac{F(\theta_l)+F(\theta_r))}{2}-F\left(\frac{\theta_l+\theta_r}{2}\right)=J_F(\theta_l,\theta_r)$ is a Jensen divergence.
Chernoff symmetrization of Bregman divergences was first proposed in~\cite{Chen-BD2-2008}. Metrization of $S_F$ symmetrized Bregman divergences was studied in~\cite{Chen-BD1-2008}.
\end{Remark}

\begin{Remark}
Notice that the right Bregman centroid $\bar\theta=\sum_i w_i\theta_i$ yields the (right) Bregman information~\cite{BD-2005}:
$$
\sum_{i=1}^n w_i\, B_F(\theta_i:\bar\theta) = J_F(\theta_1,\ldots,\theta_n;w_1,\ldots,w_n)
$$
where $J_F(\theta_1,\ldots,\theta_n;w_1,\ldots,w_n)$ is the Jensen diversity index~\cite{burbea1982entropy}:
$$
J_F(\theta_1,\ldots,\theta_n;w_1,\ldots,w_n)=\sum_i w_i F(\theta_i)-F(\sum_i w_i\theta_i),
$$
and the left Bregman centroid $\ubartheta=(\nabla F)^{-1}(\sum_i w_i \nabla F(\theta_i))$ yields the left Bregman information:
$$
\sum_{i=1}^n w_i\, B_F(\ubartheta:\theta_i)=J_{F^*}(\eta_1,\ldots,\eta_n;w_1,\ldots,w_n),
$$
which is usually different from $J_F(\theta_1,\ldots,\theta_n;w_1,\ldots,w_n)$ for non-quadratic Bregman generators $F$.

When $F(\theta)=\frac{1}{2}\theta^\top Q\theta$ for a positive-definite matrix $Q\succ 0$, we have the convex conjugate $F^*(\eta)=\frac{1}{2}\eta^\top Q^{-1}\eta$ (with $Q^{-1}\succ 0$). We have 
$\eta_i=Q^{-1}\theta_i$ and $\eta_i=Q\theta_i$.
It follows that $\bar\theta=\sum_{i=1}^n w_i\theta_i=Q^{-1}\bar\eta$
 and
$\bar\eta=\sum_{i=1}^n w_i\eta_i=Q\bar\theta$.
Thus we check that the information radii coincide when dealing with squared Mahalanobis Bregman divergences:
\begin{eqnarray}
I_F &=& \sum_{i=1}^n w_i \frac{1}{2}\theta_i^\top Q\theta_i - \frac{1}{2}\bar\theta^\top Q\bar\theta,\\
&=&  \sum_{i=1}^n w_i \frac{1}{2}(Q^{-1}\eta_i)^\top Q (Q^{-1}\eta_i) - \frac{1}{2} (Q^{-1}\bar\eta)^\top Q (Q^{-1}\bar\eta),\\
&=& \sum_{i=1}^n w_i \eta_i^\top Q^{-1}\eta_i - \frac{1}{2} \bar\eta Q^{-1}\bar\eta,\\
&=& I_{F^*}.
\end{eqnarray}
\end{Remark}

Notice that the Jeffreys divergence~\cite{nielsen2013jeffreys} between two densities $p_{\theta_1}$ and $p_{\theta_2}$ of a same exponential family with cumulant function $F$ amounts to a symmetrized Bregman divergence $S_F$: 
$$
D_J(p_{\theta_1},p_{\theta_2}):=D_\KL(p_{\theta_1}:p_{\theta_2})+D_\KL(p_{\theta_2}:p_{\theta_1})=S_F(\theta_1,\theta_2).
$$
Thus the Jeffreys centroid of a set of densities $\{p_{\theta_i}\}$ amounts to the symmetrized Bregman centroid of the parameters $\{\theta_i\}$. 

\subsubsection{Symmetrized Kullback-Leibler centroid aka. Jeffreys centroid}\label{sec:jdiv}

Consider $F(\theta)=\theta\log\theta-\theta$ (with $B_F$ corresponding to the Kullback-Leibler divergence extended to positive reals), we end up with the  following problem for the symmetrized extended KLD centroid (extended Jeffreys centroid):
$$
\min_{\theta\in\bbR_{>0}} - (a-\theta)\log\theta-\theta\log g \equiv \min_{\theta\in\bbR_{>0}} \theta\log\frac{\theta}{g}-a\log\theta,
$$
where $a=\sum_{i=1}^n w_i\theta_i$ is the weighted arithmetic mean and $g=\prod_{i=1}^n \theta_i^{w_i}$ is the weighted geometric mean.
Solving the minimization problem yields~\cite{nielsen2013jeffreys}:
\begin{equation}\label{eq:SKL1D}
\theta=\frac{a}{W\left(\frac{a}{g}e\right)}.
\end{equation}

\begin{figure}%
\centering
\begin{tabular}{cc}
\includegraphics[width=0.4\columnwidth]{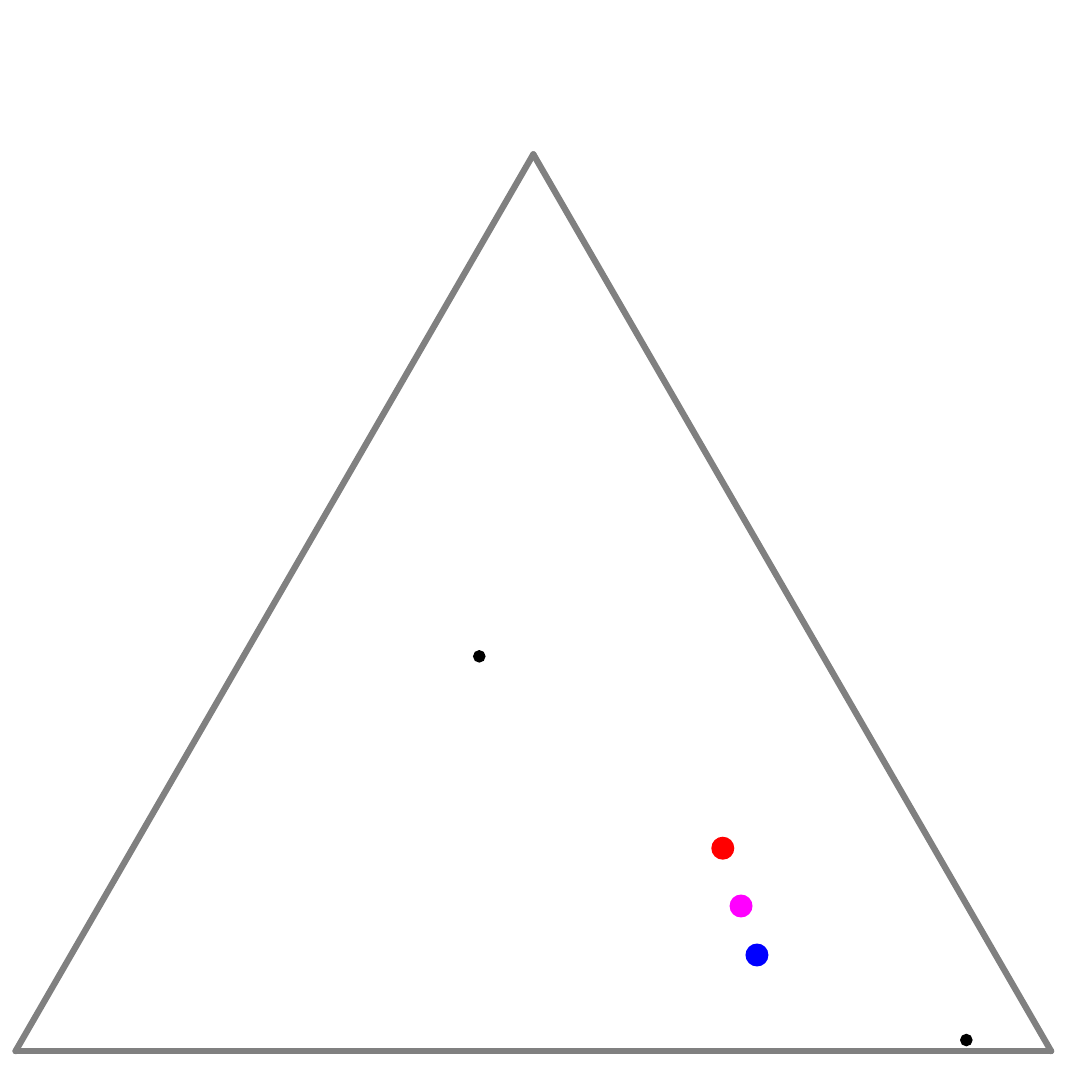} & 
\includegraphics[width=0.4\columnwidth]{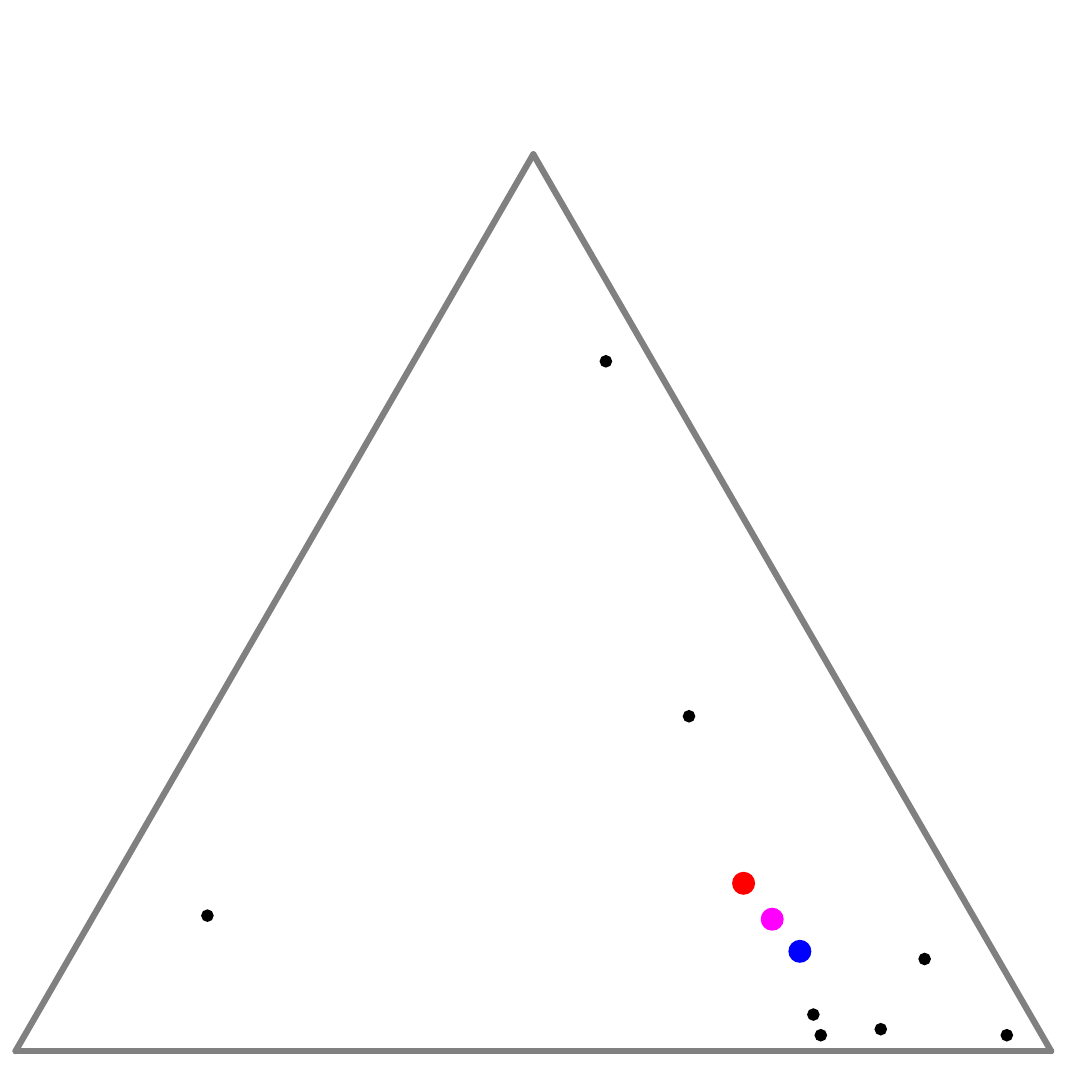}  
\end{tabular}

\caption{Centroids of categorical distributions with $3$ bins (trinomials with one trial):
Left-sided (blue) and right-sided Bregman centroids (red) correspond to the right-sided (blue,  normalized geometric mean)) and left-sided (red, arithmetic mean) Kullback-Leibler divergence, respectively. 
The symmetrized Bregman centroid corresponding to the Jeffreys centroid is shown in purple.
Left: $n=2$, Right: $n=8$.  }%
\label{fig:jeffreyshisto}%
\end{figure}

The Jeffreys centroid of a set of densities belonging to a same exponential family amounts to a symmetrized Bregman centroid.
In particular, the exact characterization of the Jeffreys centroid of categorical distributions (finite discrete distributions) has been reported in~\cite{nielsen2024fast,nielsen2013jeffreys}:
\begin{equation}\label{eq:SKLparam}
\theta_i=\frac{a_i}{W\left(\frac{a_i}{g_i}e^{1+\lambda}\right)},
\end{equation}
where $(a_i)$ and $(g_i)$ are the arithmetic and normalized geometric mean distributions, and $\lambda>0$ is unique such that the Jeffreys centroid is a normalized  probability distribution.

Thus for the particular cumulant function corresponding to the categorical exponential family, we have an exact characterization of the curved Bregman centroid.
Figure~\ref{fig:jeffreyshisto} displays the left-sided and right-sided Bregman centroids which correspond to the right-sided and left-sided Kullback-Leibler divergence, respectively. The symmetrized Bregman centroid corresponding to the Jeffreys centroid is shown in purple.

Notice that the SKL centroid for a set of non-parametric densities has similar solution to Eq.~\ref{eq:SKLparam}.
See~\cite{IG-2014} (Section~4) and related paper on non-parametric $\alpha$-centroids~\cite{amari2007integration}.

\subsubsection{Symmetrized Itakura-Saito/COSH centroid}\label{sec:cosh}

Consider $F(\theta)=-\log\theta$ (Burg negentropy with $B_F$ corresponding to the Itakura-Saito divergence~\cite{DistCentroidSound-1980}), we end up with the following optimization problem for the symmetrized Itakura-Saito centroid (also called COSH centroid):
$$
\min_{\theta\in\bbR_{>0}}  -\frac{\theta-\bartheta}{\theta}+\frac{\theta}{\ubartheta},
$$
where $\bartheta:=a$ is the weighted arithmetic mean and $\ubartheta=\frac{1}{\sum_{i=1}^n w_i \frac{1}{\theta_i}}:=h$, the weighted harmonic mean.
We find that $\theta=\sqrt{ah}=\sqrt{\bar\theta\underline{\theta}}=:g$, the weighted geometric mean.
Figure~\ref{fig:sbdis} illustrates the COSH centroid/symmetrized Itakura-Saito centroid for two parameters $\theta_1=1$ and $\theta_2=2$.

\begin{figure}
\centering

\includegraphics[width=0.8\textwidth]{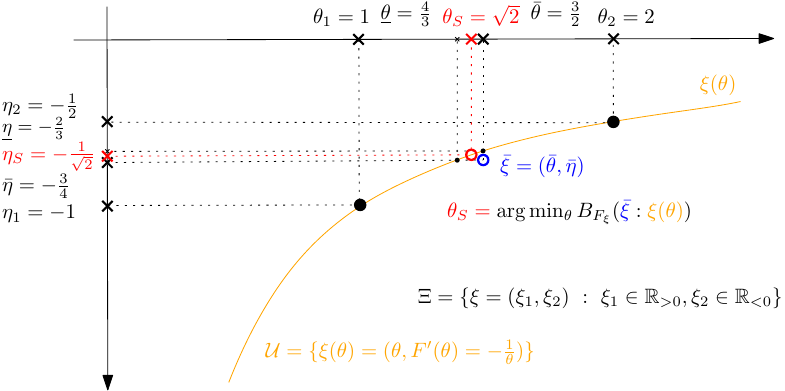}

\caption{Illustration of the symmetrized Bregman centroid for the Burg negentropy $F(\theta)=-\log\theta$ for two parameters $\theta_1=1$ and $\theta_2=2$. 
The geometric mean $\sqrt{\theta_1\theta_2}=\sqrt{2}$ is found as the right-sided Bregman projection $B_{F_\xi}$ of $\bar\xi=(\bar\theta=\frac{3}{2},\bar\eta=-\frac{3}{4})$ onto $\mathcal{U}$.
}\label{fig:sbdis}

\end{figure}

Appendix~\ref{sec:a:sbdis} provides a symbolic code to check that the symmetrized Itakura-Saito centroid is the geometric mean of the sided centroids (i.e., arithmetic and harmonic means).

\subsubsection{Symmetrized Bregman log-det divergence centroid}\label{sec:matCOSH}

\begin{figure}%
\centering
\begin{tabular}{cc}
\includegraphics[width=0.4\columnwidth]{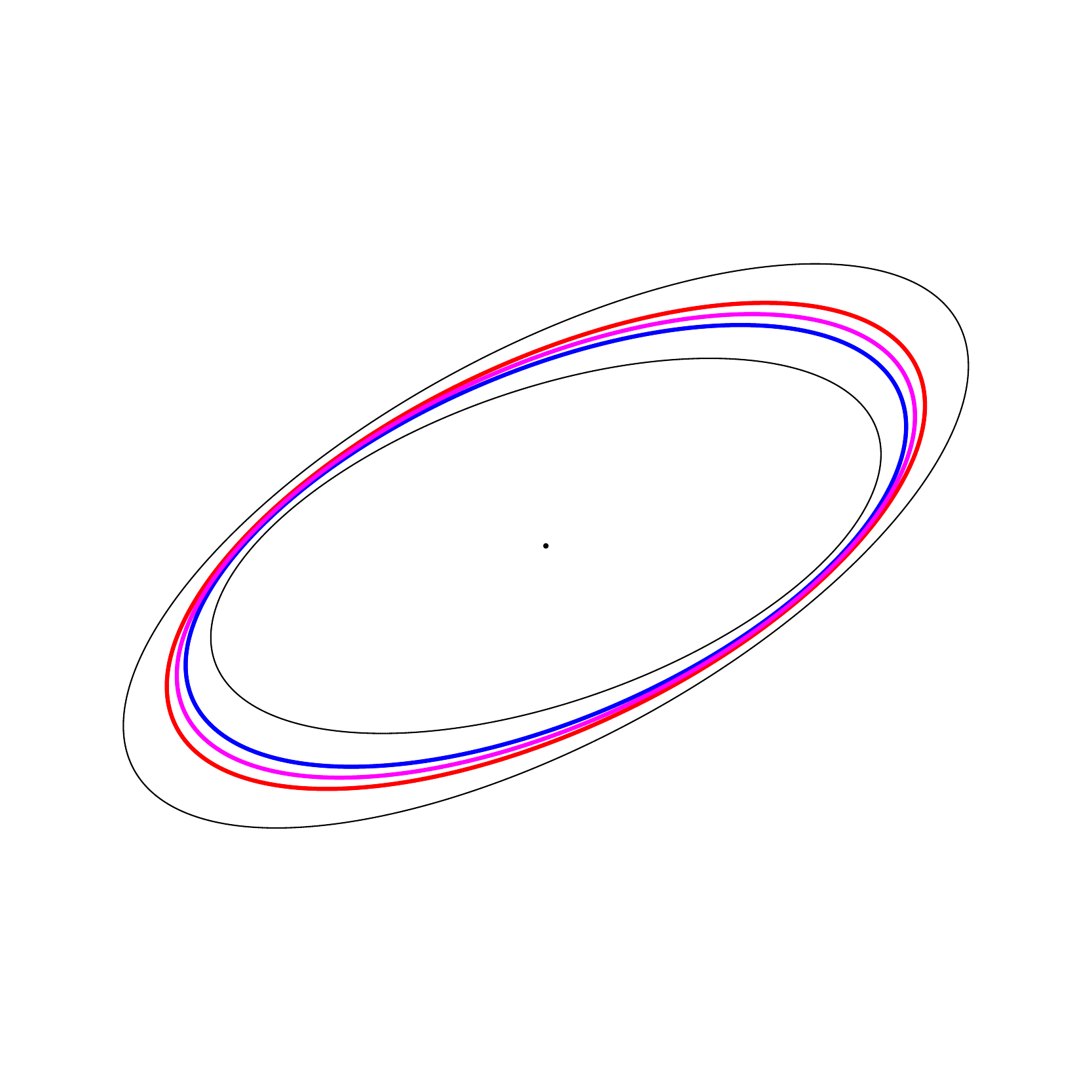} & 
\includegraphics[width=0.4\columnwidth]{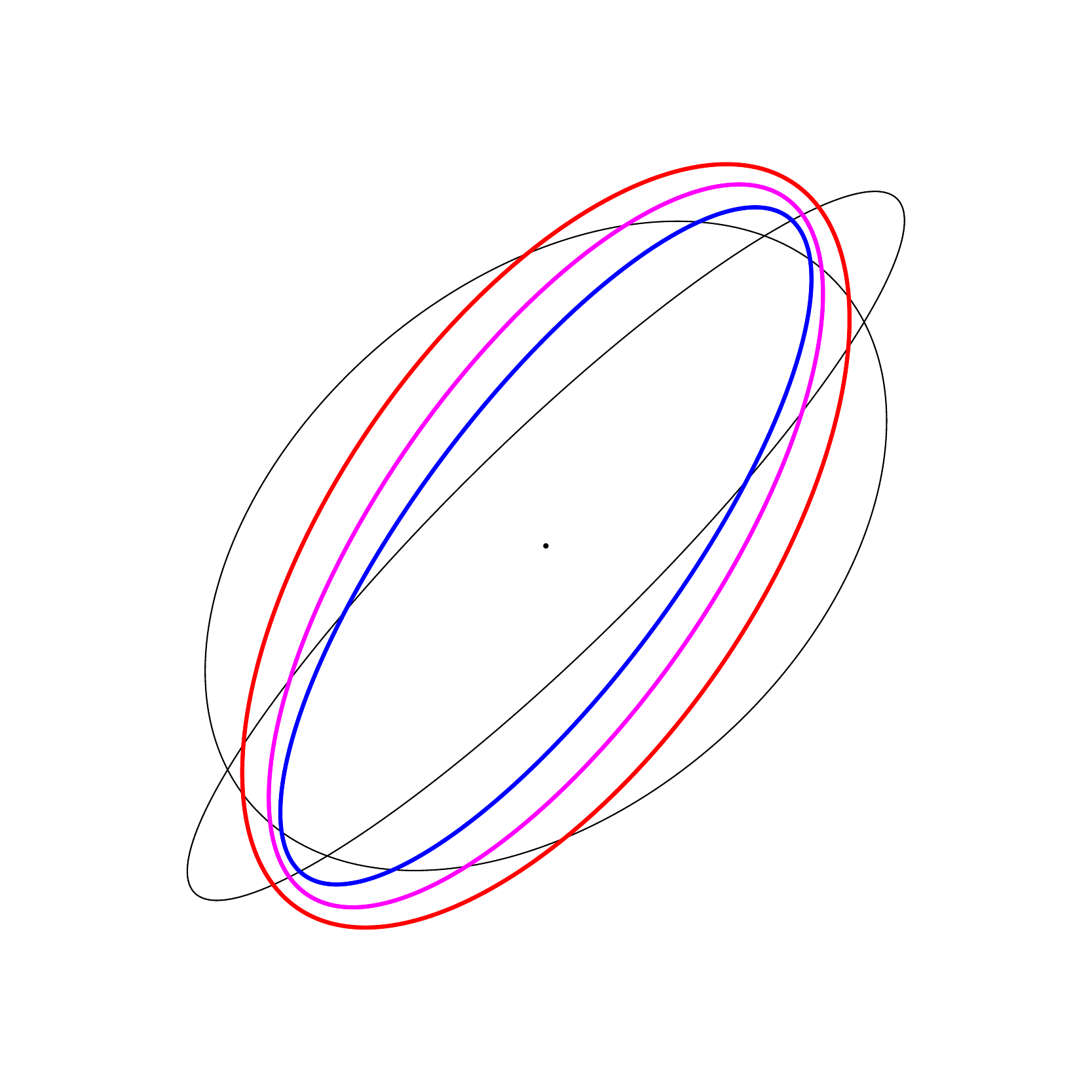}  \\
\includegraphics[width=0.4\columnwidth]{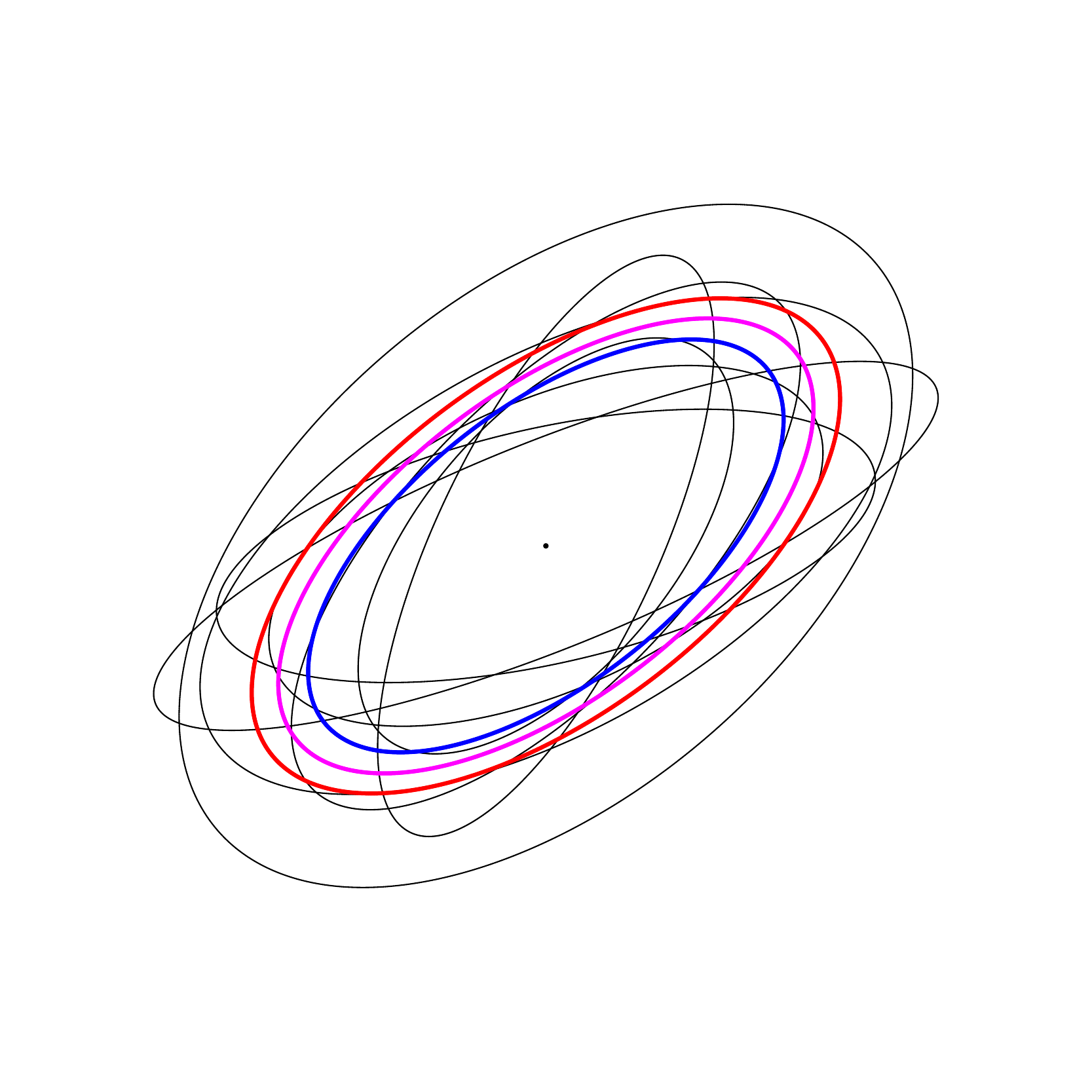} & 
\includegraphics[width=0.4\columnwidth]{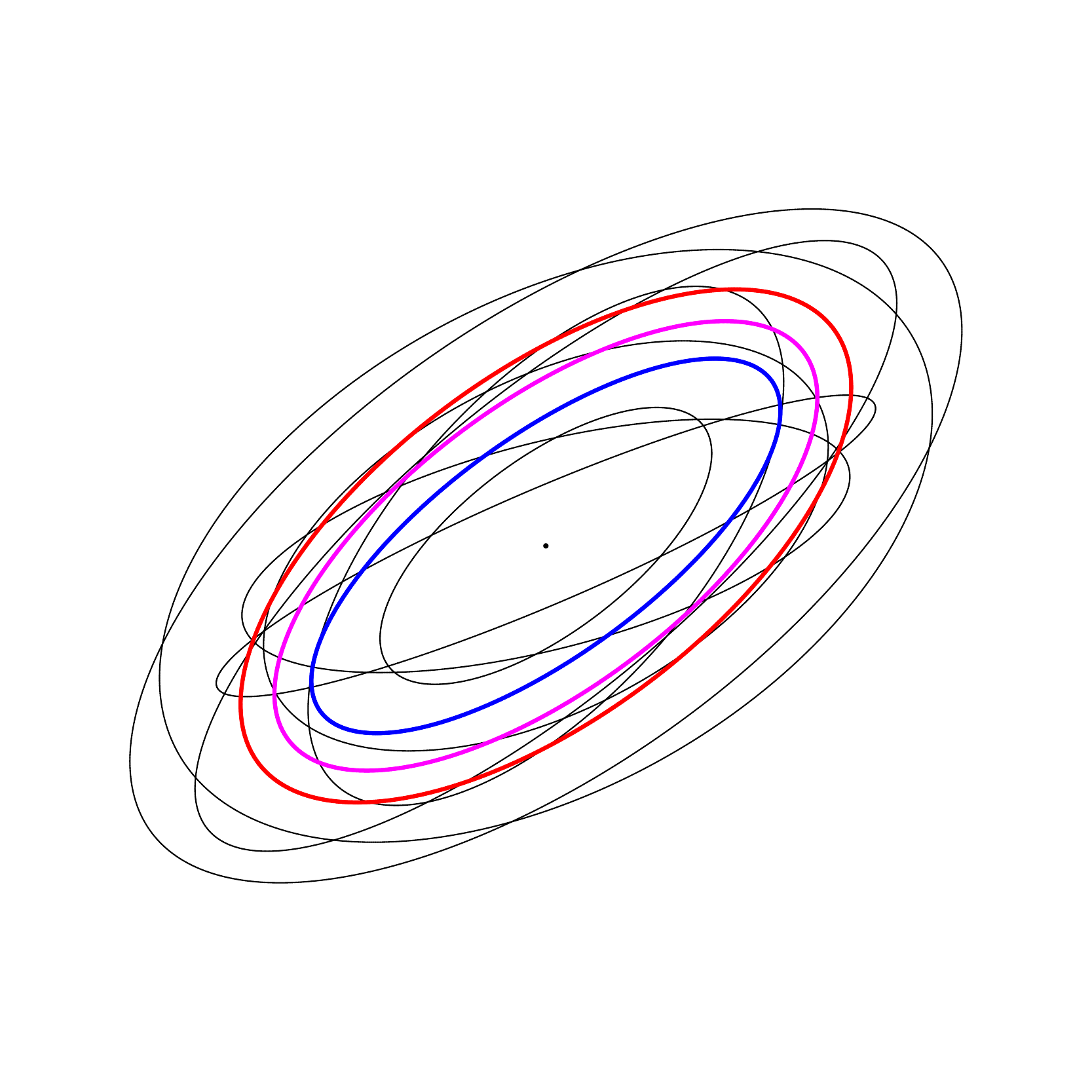}  
\end{tabular}

\caption{Left-sided (blue, harmonic matrix mean), right-sided (red, arithmetic matrix mean), and symmetrized (purple, geometric matrix mean) Bregman centroids for the log det divergence between SPD matrices displayed as corresponding centered ellipsoids. Top row: $n=2$, bottom row $n=8$.}%
\label{fig:logdetcentroid}%
\end{figure}

A matrix $X$ is symmetric positive-definite (SPD) if for all $x\not=0$, we have $x^\top X x>0$. We write $X\succ 0$ for a SPD matrix $X$ where $\succ$ is Lowner partial ordering: $A\succ B \Leftrightarrow A-B\succ 0$.
Consider the Bregman divergence on the open cone $\Theta$ of SPD matrices induced by the generator $F(\theta)=-\log\det\theta$.
This Bregman log det divergence~\cite{giraldi2018optimal} is also called the matrix Itakura-Saito divergence~\cite{ito2015blind} as it generalizes the Itakura-Saito divergence to SPD matrices.
Wang and Vemuri~\cite{SymLogDetKLClosedForm-2005} reported the following closed-form solution for the centroid with respect to the symmetrized matrix Itakura-Saito divergence (matrix COSH divergence) of a set $\theta_1, \ldots, \theta_n$ of SPD matrices:
$$
C= (\barH)^{\frac{1}{2}} \, \left((\barH)^{-\frac{1}{2}} \, \barA \, (\barH)^{-\frac{1}{2}} \right)^{\frac{1}{2}},
$$
where $\barA=\frac{1}{n}\sum_{i=1}^n \theta_i$ is the matrix arithmetic mean and 
$\barH=\left(\frac{1}{n}\sum_{i=1}^n \theta_i^{-1}\right)^{-1}$ is the matrix harmonic mean.

To a $m\times m$ SPD matrix $X\succ 0$, we associate the $m$-dimensional ellipsoid $E(X)=\{x\in\bbR^m\st x^\top X x\leq 1\}$.
See Figure~\ref{fig:logdetcentroid} for some illustrations of the sided and symmetrized logdet Bregman divergences.

Notice that the arithmetic-geometric-harmonic matrix mean inequality holds for SPD matrices~\cite{ando1979concavity,ando1983arithmetic}: This is visualized in Figure~\ref{fig:logdetcentroid} by the fact that the harmonic ellipsoid (blue) is contained in the geometric ellipsoid (purple) which is contained in the arithmetic ellipsoid (red). Indeed, $A\succ B$ iff. $E(B)\subset E(A)$.

When the matrices are 1D (ie., positive reals that commute with themselves), we recover the formula of the scalar geometric mean:
$$
c = (\bar h)^{\frac{1}{2}} \, \left((\bar h)^{-\frac{1}{2}} \, {\bar A} \, (\bar h)^{-\frac{1}{2}} \right)^{\frac{1}{2}} 
= {\bar h}^{\frac{1}{4}}  {\bar a}^{\frac{1}{2}} {\bar h}^{\frac{1}{4}} = \sqrt{{\bar a}{\bar h}} = {\bar g}.
$$


\subsection{Pointwise Bregman centroid}\label{sec:pwBDcentroid}

Consider the pointwise Bregman divergence~\cite{jones2002general} of Example~\ref{ex:pointwiseBD}.
Let $p_1,\ldots, p_n$ be $n$ weighted densities of $\calP_\mu$ with weight vector $w\in\Delta_n$.
We define the centroid constrained to a subset $\calF$ of $\calP_\mu$ to  be the minimizer of 
$$
\min_{p\in\calF} \sum_{i=1}^n w_i B_{f}^\mu(p_i:p).
$$

When $\calF=\calP_\mu$, the unconstrained centroid $c$ is found to be the center of mass~\cite{jones2002general,amari2007integration,IG-2014}:
$$
c= \sum_{i=1}^n w_i p_i.
$$

Consider the case where $\calF=\{p_\theta\st \theta\in\Theta\}$ is a subset of $\calP_\mu$ parametrized by a finite-dimensional parameter.

Then the curved/constrained Bregman centroid is defined as the minimizer of
$$
\min_{\theta\in\Theta} E_\mu(\theta) = \sum_{i=1}^n w_i B_{f}^\mu(p_i:p_\theta).
$$

Similarly to the ordinary Bregman divergence proof, we have
\begin{eqnarray*}
E_\mu(\theta) &=& -\int f(p_\theta(x))\dmu(x)-\int (\bar p-p_\theta(x)) f'(p_\theta(x))\dmu(x),\\
&=& \int \left( f(\bar p(x)) f(p_\theta(x)) - (\bar p-p_\theta(x)) f'(p_\theta(x)) \right) \dmu(x),\\
&=& B_f^\mu(\bar p:p_\theta).
\end{eqnarray*}
Thus the pointwise curved Bregman centroid amounts to the pointwise Bregman divergence projection of the mean density $\bar p$ onto the subfamily $\calF$.
See also Theorem~4 of~\cite{jones2002general}.

\begin{Theorem}[Pointwise curved Bregman centroid]
Let $p_1,\ldots, p_n$ be $n$ densities of $\calP_\mu$ and $\calF=\{p_\theta\st\theta\in\Theta\}$ a parameterized subspace of $\calP_\mu$.
Then the pointwise curved Bregman centroid amounts to the pointwise Bregman divergence projection of  $\bar p$ onto $\calF$:
$$
\arg\min_{\theta\in\Theta} \sum_{i=1}^n w_i B_{f}^\mu(p_i:p_\theta) = \arg\min_{\theta\in\Theta} B_f^\mu(\bar p:p_\theta).
$$

\end{Theorem}

Pelletier considered the case of densities $p_i=p_{\theta_i}$ to belong to a full exponential family $\calE$, and proved that 
the KL projection of the mixture $m(w)=\sum_{i=1}^n w_i p_{\theta_i}$ onto $\calE=\{p_\theta\st\theta\in\Theta\}$ is unique:
$$
\arg\min_{\theta\in\Theta} D_\KL(m(w):p_\theta)= \sum_{i=1^n} w_i\theta_i :=\bar\theta.
$$
See also~\cite{schwander2012learning} for another proof. 

\def\sqr{\mathrm{sqr}}
\subsection{Additively weighted quadratic Bregman divergences (AWQ-BDs)}\label{sec:addweightquadBD}

One problem with the Jeffreys-type symmetrization $S_F$ of the Bregman divergence $B_F$ is that $\sqrt{S_F}$ may not be a metric distance.
See~\cite{chen2008metrics} for a study of other symmetrizations of Bregman divergences with conditions yieldings to metrics.

Now, consider the following symmetrization reported in~\cite{acharyya2013bregman} of any $m$-dimensional Bregman divergence $B_F(\theta_2:\theta_1)$ for $\theta_1,\theta_2\in\Theta\subset\bbR^m$ and $F$ a Legendre-type generator (with gradient space $H=\{\nabla F(\theta) \st \theta\in\Theta\}$):

\begin{eqnarray}
S_F^{(\alpha,\beta)}(\theta_1,\theta_2) &=& \inner{\theta_2-\theta_1}{\nabla F(\theta_2)-\nabla F(\theta_1)}\nonumber\\
&& +\frac{\alpha}{2} \inner{\theta_2-\theta_1}{\theta_2-\theta_1}+\frac{\beta}{2} \inner{\nabla F(\theta_2)-\nabla F(\theta_1)}{\nabla F(\theta_2)-\nabla F(\theta_1)},
\end{eqnarray}
for any $\alpha\geq 0$ and $\beta\geq 0$.
When $\alpha=\beta=0$, we recover the Jeffreys symmetrization of Bregman divergences: $S_F^{(0,0)}(\theta_1,\theta_2)=S_F(\theta_1,\theta_2)$.
Note that $S_F^{(\alpha,\beta)}(\theta_1,\theta_2)=S_{F^*}^{\beta,\alpha}(\eta_1,\eta_2)$ since ${F^*}^*=F$.

Now, when $\alpha\beta\geq 1$, $\sqrt{S_F^{(\alpha,\beta)}}$ is proved to be a metric distance (Theorem~3 of \cite{acharyya2013bregman}).
Furthermore,  $\sqrt{S_F^{(\alpha,\beta)}}$ is shown to be a Hilbert metric distance using the following feature mapping $\Phi: \theta\rightarrow\Theta\times H\subset\bbR^{2m}$ (Lemma~1 of \cite{acharyya2013bregman}):

\begin{eqnarray}
\Phi_{\alpha,\beta}(\theta) &=&
 \left[\begin{array}{ll}
\sqrt{\alpha}\theta+\frac{1}{\sqrt{\alpha}}\nabla F(\theta)\\
\sqrt{\frac{\alpha\beta-1}{\alpha}}\, \nabla F(\theta)
\end{array}\right]\in\bbR^{2m},\\ 
S_F^{(\alpha,\beta)}(\theta_1,\theta_2) &=& \frac{1}{2}\, \|\Phi_{\alpha,\beta}(\theta_1)-\Phi_{\alpha,\beta}(\theta_2)\|_2^2,\quad \alpha>0,\beta>0,\alpha\beta>1.
\end{eqnarray}

Notice that the right-hand-side of the last equation is a squared Euclidean distance, namely a Bregman divergence:
$$
S_F^{(\alpha,\beta)}(\theta_1,\theta_2) = B_{F_\sqr}(\Phi(\theta_1),\Phi(\theta_2)),
$$
where $F_\sqr(x,y)=\inner{x}{x}+\inner{y}{y}=\|x\|^2+\|y\|^2$.

Thus the AWQ-BDs are curved squared Euclidean divergence (and more generally, kernel methods in machine learning yield curved Euclidean distances in RKHS spaces).
Except for quadratic generators $F(\theta)=A\theta$ for a symmetric positive-definite (SPD) matrix $A$, those additively weighted quadratic Bregman divergences (AWQ-BDs) are not Bregman divergences.

Consider the generator
$$
\hat{F}_{\alpha,\beta}(\xi_1,\xi_2) = F(\xi_1)+\frac{\alpha}{2} \inner{\xi_1}{\xi_1} + F^*(\xi_2)+ \frac{\beta}{2}\inner{\xi_2}{\xi_2},
$$ 
for $\xi=(\xi_1,\xi_2)$. This is a Legendre-type function yields a Bregman divergence $B_{\hat{F}_{\alpha,\beta}}$ with
$$
\nabla\hat{F}_{\alpha,\beta}(\xi_1,\xi_2)=
\left[\begin{array}{ll}
\nabla F(\xi_1) + \alpha\xi_1\\
\nabla F^*(\xi_2) + \beta\xi_2
\end{array}\right].
$$

Let $\xi(\theta)=(\theta,\nabla F(\theta))$ be the embedding map.
Then we have
\begin{equation}
S_F^{(\alpha,\beta)}(\theta_1,\theta_2) = B_{\hat{F}_{\alpha,\beta}}(\xi(\theta_1),\xi(\theta_2)).
\end{equation}
Thus $S_F^{(\alpha,\beta)}$ is a $m$-dimensional curved Bregman divergence of $\bbR^{2m}$ where $\calU=\{\xi(\theta) \st \theta\in\Theta\}$.

Now, the centroid of weighted $n$ points $\theta_1,\ldots,\theta_n$ (with weight vector $w\in\Delta_n$) with respect to $S_F^{(\alpha,\beta)}$ for $\alpha\beta>1$
 amounts to squared Euclidean/curved  Bregman projections of either the Euclidean centroid on the corresponding set of features $\phi_1=\Phi_{\alpha,\beta}(\theta_1),\ldots, \phi_n=\Phi_{\alpha,\beta}(\theta_n)$ or on the embedded parameters $\xi(\theta_1),\ldots, \xi(\theta_n)$:

\begin{eqnarray*}
&&\arg\min_{\theta\in\Theta} \sum_{i=1}^n w_i \, S_F^{(\alpha,\beta)}(\theta,\theta_i),\\
&&\arg\min_{\theta\in\Theta} \sum_{i=1}^n w_i\, \|\Phi(\theta)- \bar\phi\|_2^2,\\
&&\arg\min_{\theta\in\Theta} \sum_{i=1}^n w_i\, B_{\hat{F}}(\bar\xi:\xi(\theta)),
\end{eqnarray*}
where  $\bar\phi=\sum_{i=1}^n w_i\,\phi_i$ and $\bar\xi=\sum_{i=1}^n w_i\,\xi(\theta_i)$.

%
%

\begin{Remark}
Notice that when $\alpha\beta=1$ and $\alpha>0$, the mapping $\Phi$ is from $\Theta\rightarrow\Theta$ since the second vector block vanishes in that particular case.
To fix ideas, consider $m=1$, and $\alpha=\beta=1$. Then we have the diffeomorphic mapping $\Phi(\theta)=\theta+F'(\theta)$.
$\Phi$ is an increasing function corresponding to the derivative of the Legendre type function $G(\theta)=F(\theta)+\frac{1}{2}\theta^2$.
Thus the inverse function $(G')^{-1}$ is well-defined globally.
\end{Remark}

In~\cite{acharyya2013bregman}, the AWQ-BDs are further generalized using two SPD matrices $A$ and $B$ as follows:

\begin{equation}
S_F^{(A,B)}(\theta_1,\theta_2) = S_F(\theta_1,\theta_2)
  + (\theta_2-\theta_1)^\top A (\theta_2-\theta_1)+
	(\nabla F(\theta_2)-\nabla F(\theta_1))^\top B (\nabla F(\theta_2)-\nabla F(\theta_1)).
\end{equation}
It is shown that $\sqrt{S_F^{(A,B)}}$ is a metric provided that $AB-I$ is SPD (Theorem~4 of~\cite{acharyya2013bregman}).
When $A=\alpha I$ and $B=\beta I$, we have $AB-I=(\alpha\beta-1)I$ which is SPD when $\alpha\beta>1$ and PD when $\alpha\beta=1$.

\subsection{Kullback-Leibler divergence between circular complex normal distributions}\label{sec:CN}
The $d$-variate circular complex normal distribution~\cite{ollila2012complex}  $\CN_d(\mu_\bbC,S_\bbC)$ can be handled as a 
$2d$-real normal distribution $N_{2d}([\mu_\bbC]_\bbR,\frac{1}{2} [S_\bbC]_\bbR)$ where $[z=a+ib]_\bbC=(a,b)$ and $[M=A+iB]_\bbC=\mattwo{A}{-B}{B}{A}$, where $a,b\in\bbR$ and $A,B\in\bbR^{d\times d}$.
Let $N(m_\bbC,S_\bbC)$ and $N(m_\bbC',S_\bbC')$ be two circular complex normal distributions.
The Kullback-Leibler divergence between those two distributions amounts to a KLD between their 
corresponding real normal distributions $N_{2d}(\mu,\Sigma)$ and $N_{2d}(\mu',\Sigma')$ where $\mu=[m_\bbC]_\bbC$, $\Sigma=[S_\bbC]_\bbC$, 
and $\mu'=[m'_\bbC]_\bbC$, $\Sigma'=[S_\bbC']_\bbC$.
Since the KLD between two densities of an exponential family amounts to a reverse Bregman divergence between their natural parameters~\cite{azoury2001relative}, 
we have
\begin{equation}
D_\KL[p_{m_\bbC,S_\bbC}:p_{m_\bbC',S_\bbC'}] = D_\KL[p_{\mu,\Sigma}:p_{\mu',\Sigma'}] = B_F(\theta':\theta).
\end{equation}
The Gaussian family $\{N(\mu,\Sigma)\}$ is an exponential family for natural parameter $\theta=(\Sigma^{-1}\mu,\frac{1}{2}\Sigma^{-1})$ and cumulant function (or log-partition function) $F(\theta)=F(\theta_1,\theta_2)=\frac{1}{2}\left(d\log\pi -\log\det(\theta_2)+\frac{1}{2} \theta_1^\top \theta_2^{-1}\theta_1\right)$.
Thus the family of $d$-variate circular complex normal distributions is  a curved exponential family of $2d$-variate real normal distributions with
$$
\calU=\left\{ (v,M) \st v\in\bbR^{2d}, M=\mattwo{A}{-B}{B}{A}, A\in\bbR^{d\times d}\succ 0, B\in\bbR^{d\times d}\succ 0\right\}.
$$

More generally, any matrix over a finite-dimensional real algebra can be represented as a real matrix~\cite{MatrixGroupBaker-2003}.

When $\Theta_{\calU}=\{u(\theta) \st \theta\in\Theta\}$ is the restriction of a natural parameter space to an affine subspace, the curved Bregman divergence is a sub-dimensional Bregman divergence.

\begin{Example}\label{ex:curvedKLDcat}
The Kullback-Leibler divergence (KLD) $D_\KL^+$ extended to positive arrays~\cite{Csiszar-1991} $\Theta=\bbR_{>0}^m$ (positive orthant cone)
 is a separable Bregman divergence obtained for the generator $F_{\KL^+}(u)=\sum_{i=1}^m u_i\log u_i$: 
$D_\KL^+(\theta:\theta')=\sum_{i=1}^m \theta_i\log\frac{\theta_i}{\theta'_i}+\theta_i'-\theta_i$.
Consider the hyperplane $H_\Delta:\sum_{i=1}^m x_i=1$. Then $\Theta_{|H_\Delta}=\Delta_m$ is the standard simplex, 
and $u(x)=(x_1,\ldots,x_{m-1})\in\Delta_{m-1}$ is the moment parameter of the simplex exponential family.
Then the curved Bregman divergence $B_{F_{\KL^+}}(u(\theta):u(\theta'))$ which corresponds to the KLD defined on the standard simplex amounts to a sub-dimensional Bregman divergence: 
$B_{F_{\KL^+}}(u(\theta):u(\theta'))=B_{F_\KL}(\alpha:\alpha')$ with $\alpha=u(\theta)$, $\alpha'=u(\theta')$ in $\Delta_m$ for the non-separable generator 
$F_\KL(\alpha)=\sum_{i=1}^{m-1} \alpha_i\log\alpha_i+(1-\sum_{i=1}^{m-1} \alpha_i)\log(1-\sum_{i=1}^{m-1} \alpha_i)$. 
\end{Example}

\subsection{Generalized left-sided Bregman centroids}\label{sec:leftsidedBDcentroid}

Recall that the left-sided Bregman centroid~\cite{nielsen2009sided} of a weighted parameter set $\{(w_1,\theta_1), \ldots, (w_n,\theta_n)\}$ with normalized weight vector $w\in\Delta_n$ and the $\theta_i$'s in $\Theta$ (with $\dim(\Theta)=m$) is the unique minimizer of the following objective function:

$$
\sum_{i=1}^n w_i \, B_F(\theta:\theta_i).
$$

Let us consider the $n$-fold Bregman divergence induced by the generator $\bar F: \Theta^n\rightarrow \bar\bbR$ (with $\dim(\Theta^n)=nm$):
$$
\bar F(\xi)=\sum_{i=1}^n w_i \, F(\xi_i),
$$
where $\xi=(\xi_1,\ldots, \xi_n)\in \Theta^n$.

We may further consider the more general left-sided Bregman centroid by allowing a different Bregman generator $F_i:\Theta_i\rightarrow \bar\bbR$ for each parameter  and ask to minimize the following criterion:
\begin{equation}\label{eq:genlsc}
\sum_{i=1}^n w_i\, B_{F_i}(\theta:\theta_i).
\end{equation}

We shall now use the property that $\alpha B_F(\theta:\theta')+\beta B_G(\theta:\theta')=B_{\alpha F+\beta G}(\theta:\theta')$ for any $\alpha>0, \beta>0$ and $F$ and $F$ strictly convex and differentiable convex functions.

Thus we define the $n$-fold Bregman generator by $\bar F: \Theta_{(n)}=\prod_{i=1}^n\Theta_i\rightarrow \bar\bbR$
$$
\bar F(\xi) = \sum_{i=1}^n w_i F_i(\xi_i),
$$
for $\xi=(\xi_1,\ldots,\xi_n)\in \Theta_{(n)}$.

\def\diag{\mathrm{diag}}

The domain $\Theta_{(n)}$ is convex since for any $\lambda\in [0,1]$, we have $\lambda\xi+(1-\lambda)\xi'=(\lambda\xi_1+(1-\lambda)\xi'_1,\ldots,\lambda\xi_n+(1-\lambda)\xi')\in \Theta_{(n)}$ since $\lambda\xi_i+(1-\lambda)\xi'_i\in\Theta_i$, 
 and generator $\bar F$ is a strictly convex Bregman generator since
 $\nabla^2_i\bar F(\xi)=\diag(w_1 \nabla^2 F_1(\xi_1), \ldots, w_n \nabla^2 F_1(\xi_n))$ is positive-definite (blockwise positive-definite Hessian matrix).

Next, assume that all $\Theta_i$'s coincide to $\Theta$ so that $\Theta_{(n)}=\Theta^n$.
Let $\theta_{(n)}=(\theta_1,\ldots,\theta_n)\in\Theta_{(n)}$ and consider the mapping $\xi(\cdot): \Theta\rightarrow \Theta^n$, $\theta\mapsto\xi(\theta)=(\theta,\ldots,\theta)\in\Theta^n$.
Then the optimization problem of Eq.~\ref{eq:genlsc} can be written  equivalently as:

$$
\min_{\theta\in\Theta} B_{\bar F}(\xi(\theta):\theta_{(n)}) = \sum_{i=1}^n w_i\, B_{F_i}(\theta:\theta_i).
$$

That is, the generalized left-sided Bregman centroid amounts to a left-sided Bregman projection of the $n$-fold input set compound parameter $\theta_{(n)}$ of $\Theta^n$ onto
 the $m$-dimensional submanifold $\{\xi(\theta) \st\theta\in\Theta\}\subset \Theta_{(n)}$ with $\dim(\Theta_{(n)})=nm$.

\begin{Proposition}\label{prop:leftsidedGenBregmanCentroid}
The generalized left-sided Bregman centroid $\theta_L$ of a weighted point set $\{(w_1,\theta_1), \ldots, (w_n,\theta_n)\}$ with $w\in\Delta_n$ and $\theta_i\in\Theta$ minimizing $\sum_{i=1}^n w_i\, B_{F_i}(\theta:\theta_i)$ for $F_i:\Theta\rightarrow\bar\bbR$ amounts to a left-sided Bregman projection with respect to the compound Bregman divergence $B_{\bar F}$:
 $\min_{\theta\in\Theta} B_{\bar F}(\xi(\theta):\theta_{(n)})$.
Furthermore, we have $\theta_L=(\nabla\bar F)^{-1} \left(\sum_i w_i\nabla F_i(\theta_i)\right)$ where $\nabla\bar F(\xi)=\sum_i w_i\nabla F_i(\xi_i)$.
\end{Proposition}

\begin{proof}
Let us write $B_{F_i}$ using the equivalent Fenchel-Young divergence
$B_{F_i})(\theta:\theta_i)=F_i(\theta)+F_i^*(\nabla F_i(\theta_i))-\inner{\theta}{\nabla F_i(\theta_i)}$.
We have
\begin{eqnarray*}
&& \min_\theta \sum_{i=1}^n w_i\, B_{F_i}(\theta:\theta_i),\\
&& \min_\theta  \sum_{i=1}^n w_i\,  \left( F_i(\theta)+F_i^*(\nabla F_i(\theta_i))-\inner{\theta}{\nabla F_i(\theta_i)}\right),\\
&& \equiv \min_\theta E(\theta)={\bar F}(\theta)-\inner{\theta}{\sum_{i=1}^n w_i \nabla F_i(\theta_i)}.
\end{eqnarray*}

The minimum is found for $\nabla E(\theta)=0$, i.e., $\nabla {\bar F}(\theta)=\sum_{i=1}^n w_i \nabla F_i(\theta_i)$.
That is, for $\theta=(\nabla\bar F)^{-1} \left(\sum_i w_i\nabla F_i(\theta_i)\right)$.
Furthermore, the minimum is unique since $\nabla^2 E(\theta)=\nabla^2{\bar F}(\theta)$ is positive-definite.
\end{proof}

\section{Approximating Jeffreys-Bregman centroids}\label{sec:approxsymBD}

Let us define the scaled skewed Jensen divergence as
$$
\barJ_{F,\alpha}(\theta_l:\theta_r) = \frac{1}{\alpha(1-\alpha)} \left(\alpha F(\theta_l)+(1-\alpha) F(\theta_r) - F(\alpha\theta_l+(1-\alpha)\theta_r)\right).
$$
We have $\lim_{\alpha\rightarrow 0} \barJ_{F,\alpha}(\theta_l:\theta_r)=B_F(\theta_l:\theta_r)$ and 
$\lim_{\alpha\rightarrow 1} \barJ_{F,\alpha}(\theta_l:\theta_r)=B_F(\theta_r:\theta_l)$.

Let us design a family of symmetric $\alpha$-Jensen divergences~\cite{nielsen2010family} for $\alpha\in(0,1)$:

$$
\barJ_{F,\alpha}^s(\theta_l,\theta_r)= \barJ_{F,\alpha}(\theta_l:\theta_r) + \barJ_{F,\alpha}(\theta_r:\theta_l) 
$$
which satisfies
\begin{eqnarray*}
\lim_{\alpha\rightarrow 0} \barJ_{F,\alpha}^s(\theta_l,\theta_r) &=& S_F(\theta_l,\theta_r),\\
\lim_{\alpha\rightarrow 1} \barJ_{F,\alpha}^s(\theta_l,\theta_r) &=& S_F(\theta_l,\theta_r).
\end{eqnarray*}

Consider a small prescribed real-value  for $\alpha=\epsilon>0$.
Then the  $\alpha$-Jensen  centroid for a weighted parameter set $\theta_1,\ldots, \theta_n$ with $w\in\Delta_n$ minimizes
\begin{eqnarray*}
&&\min_{\theta\in\Theta} \sum_{i=1}^n w_i \barJ_{F,\alpha}^s(\theta_i,\theta),\\
&& \equiv \min_{\theta\in\Theta} \sum_{i=1}^n  w_i (F(\theta_i)+F(\theta)-F(\epsilon\theta_i+(1-\epsilon)\theta)-F((1-\epsilon\theta_i)+\epsilon\theta)),\\
&&  \equiv \min_{\theta\in\Theta} F(\theta)-\sum_{i=1}^n w_i\, \left( F(\epsilon\theta_i+(1-\epsilon)\theta)+ F((1-\epsilon\theta_i)+\epsilon\theta)\right).
\end{eqnarray*}

This optimization problem is a difference of a convex function $F(\theta)$ and a concave term $-\sum_{i=1}^n \left( F(\epsilon\theta_i+(1-\epsilon)\theta)+ F((1-\epsilon\theta_i)+\epsilon\theta)\right)$. We can thus apply the Convex-ConCave minimization procedure~\cite{yuille2002cccp}:
Initialize arbitrary $\theta^{(1)}_\epsilon=\theta_1$ (or using the center of mass $\sum_{i=1}^n  w_i\theta_i$),  
and iterate
$$
\theta^{(t+1)}_\epsilon = (\nabla F)^{-1}\left(
\sum_{i=1}^n  w_i\, \left( \nabla F(\epsilon\theta_i+(1-\epsilon)\theta^{(t)})+ \nabla F((1-\epsilon\theta_i)+\epsilon\theta^{(t)}_\epsilon)\right).
 \right)
$$

When $\epsilon\rightarrow 0$ and $T\rightarrow\infty$, the symmetric $\epsilon$-Jensen divergence tend to the symmetrized Bregman centroid.
For the case of the log det Bregman generator, we have $\nabla F(X)=(\nabla F)^{-1}(X)=-X^{-1}$.
For the symmetric KL case~\cite{nielsen2013jeffreys}, we use $\nabla F(x)=\frac{e^{x_i}}{1+\sum_{j=0}^{d-1} e^{x_j}}$ and 
$\nabla F^{-1}(x)=\log \frac{x_i}{1-\sum_{j=0}^{d-1} x_j}$ with
 $\theta= \left(\log\frac{p_1}{1-\sum_{j=1}^{d-1}p_j},\ldots, \log\frac{p_{d-1}}{1-\sum_{j=1}^{d-1}p_j}\right)$.
The CCCP iterations allow us to approximate finaly the symmetrized Bregman centroid without using Lambert $W$ function to finely approximate the SKL centroid~\cite{nielsen2013jeffreys}.

Observe that we have
\begin{eqnarray*}
S_F(\theta_l,\theta_r) &=& B_F(\theta_l:\theta_r)+B_F(\theta_r:\theta_l),\\
&=& B_{F^*}(\eta_r:\eta_l)+B_{F^*}(\eta_l:\eta_r)= S_{F^*}(\eta_l,\eta_r).
\end{eqnarray*}

Thus the symmetrized Bregman centroid can also be defined as
$$
\arg\min_{\eta} \sum_{i=1}^n w_i\, S_{F^*}(\eta_i,\eta),
$$
and we may approximate it by considering the $\alpha$-Jensen  centroid $\barJ_{F^*,\alpha}$  for a weighted parameter set $\eta_1,\ldots, \eta_n$ and $\alpha=\epsilon$, a prescribed small value.

The dual CCCP method is:
Initialize arbitrary $\eta^{(1)}_\epsilon=\eta_1$ (or using the center of mass $\sum_{i=1}^n  w_i\eta_i$),  
and iterate
$$
\eta^{(t+1)}_\epsilon = (\nabla F)^{-1}\left(
\sum_{i=1}^n  w_i\, \left( \nabla F(\epsilon\eta_i+(1-\epsilon)\eta^{(t)})+ \nabla F((1-\epsilon\eta_i)+\epsilon\eta^{(t)}_\epsilon)\right).
 \right).
$$

We may now alternatively perform one iteration using the $\theta$-parameterization and one parameterization using the $\eta$-parameterization to get the mixed parameterization CCCP procedure:

\begin{itemize}
\item Let $t=0$ and $\theta^{(0)}_\epsilon=\theta_1$ (or using the center of mass $\sum_{i=1}^n  w_i\theta_i$) 

\item iterate $T$ rounds:
\begin{itemize}

\item CCCP iteration on $\theta$-parameterization:
$$
\theta^{(2t+1)}_\epsilon = (\nabla F)^{-1}\left(
\sum_{i=1}^n  w_i\, \left( \nabla F(\epsilon\theta_i+(1-\epsilon)\theta^{(2t)})+ \nabla F((1-\epsilon\theta_i)+\epsilon\theta^{(2t)}_\epsilon)\right).
 \right)
$$

\item Convert $\theta$- to $\eta$-parameter: Let $\eta^{(2t+1)}=\nabla F(\theta^{(2t+1)})$.

\item CCCP iteration on $\eta$-parameterization:

$$
\eta^{(2t+2)}_\epsilon = (\nabla F)^{-1}\left(
\sum_{i=1}^n  w_i\, \left( \nabla F(\epsilon\eta_i+(1-\epsilon)\eta^{(2t+1)})+ \nabla F((1-\epsilon\eta_i)+\epsilon\eta^{(2t+1)}_\epsilon)\right).
 \right).
$$

\item Convert $\eta$- to $\theta$-parameter: Let $\theta^{(2t+2)}=(\nabla F)^{-1}(\eta^{(2t+2)})$.

\item Let $t\leftarrow t+1$ (next round).

\end{itemize}

\item Return $\theta^{(2T)}$ as an approximation of the symmetrized Bregman centroid.
\end{itemize}

%
%
%
%

\section{Curved representational Bregman divergences}\label{sec:crbd}
\subsection{Definition}

\begin{Definition}\label{def:crbd}
Let $B_F$ be a   Bregman divergence for the Legendre-type generator $F:\Theta\subset V\rightarrow \bbR$.
Consider a diffeomorphism $R: \Theta\rightarrow\Theta, \theta\mapsto r=R(\theta)$ called the representation function.
Then $B_F(R(\theta_1):R(\theta_2))=B_F(r_1:r_2)$ is termed a representational Bregman divergence~\cite{nielsen2009dual}.
\end{Definition}

Note that when $R$ is the identity function (i.e., $R(\theta)=\theta$), representational Bregman divergences become ordinary Bregman divergences.
 Examples of representational Bregman divergences are met when dealing with the Kullback-Leibler divergence (KLD) between two densities of a generic exponential family~\cite{barndorff2014information}, not necessarily a natural exponential family~\cite{azoury2001relative,nielsen2022statistical}.
Let $(\calX,\calF,\mu)$ be a measure space and $\calE=\{p_\lambda(x)=\exp(\inner{t(x)}{\theta}(\lambda))-F(\theta(\lambda))\dmu\}$ an exponential family with sufficient statistic vector  $t(x)$, natural parameter $\theta(\lambda)$ and cumulant/log-normalizer/log-partition function $F(\theta)$.
It is well-known that the KLD $D_\KL(p_{\lambda_1}:p_{\lambda_2})=\int p_{\lambda_1}(x)\log\frac{p_{\lambda_1}(x)}{p_{\lambda_2}(x)}\dmu$ between two densities $p_{\lambda_1}$ and $p_{\lambda_2}$ of an exponential family   amounts to a reverse representational Bregman divergence:
$$
D_\KL(p_{\lambda_1}:p_{\lambda_2})=B_F(\theta(\lambda_2):\theta(\lambda_1)).
$$
Here, the natural parameterization $\theta(\lambda)$ plays the role of the representation function.

\begin{Proposition}\label{prop:alphabdrep}
The $\alpha$-divergences~\cite{IG-2016} $D_\alpha^+(q_1:q_2)$ extended to the positive measures $q_1$ and $q_2$ of $\bbR_{>0}^m$:
$$
D_\alpha^+(q_1:q_2) = \left\{
\begin{array}{ll}
\frac{4}{1-\alpha^2}\, \sum_{i=1}^m \left( \frac{1-\alpha}{2}q_1 + \frac{1+\alpha}{2}q_2 - q_1^{\frac{1-\alpha}{2}}\, q_2^{\frac{1+\alpha}{2}} \right), \alpha\in\bbR\backslash\{-1,1\}\\
{D_\KL^\drev}^+(q_1:q_2) = D_\KL^+(q_2:q_1)=\sum_{i=1}^m q_2^i\log\frac{q_2^i}{q_1^i}+q_1^i-q_2^i & \alpha=1\\
D_\KL^+(q_1:q_2)= \sum_{i=1}^m q_1^i\log\frac{q_1^i}{q_2^i}+q_2^i-q_1^i & \alpha=-1.
\end{array}
\right.
$$
 is a representational Bregman divergence:
\begin{equation}
D_\alpha^+(q_1:q_2) = B_{F_\alpha}(R_\alpha(q_1):R_\alpha(q_2)) = B_{F_{-\alpha}}(R_{-\alpha}(q_2):R_{-\alpha}(q_1)),  
\end{equation}
where $F_\alpha(r)=\sum_{i=1}^m \,f_\alpha(r_i)$ is the separable Bregman generator on the representations induced by the potential function  
$$
f_\alpha(x)=\left\{
\begin{array}{ll}
\frac{2}{1+\alpha}\,\left( \frac{1-\alpha}{2} x \right)^{\frac{2}{1-\alpha}}, &\alpha\not =1\\
x\exp(x), &\alpha=1.
\end{array}
\right.
$$ 
and $r=R_\alpha(q)=(r_\alpha(q_1),\ldots,r_\alpha(q_m))$ is the representation function induced by 
$$
r_\alpha(x)=\left\{\begin{array}{ll}
\frac{2}{1-\alpha} \left( x^{\frac{1-\alpha}{2}}-1\right), & \alpha\not =1,\\
\log x, & \alpha=1.
\end{array}
\right.
$$


Furthermore, $D_\alpha^+$ can be expressed as an equivalent representational Fenchel-Young divergence:
\begin{equation}
D_\alpha^+(q_1:q_2) = Y_{F_\alpha,F_{-\alpha}}(R_\alpha(q_1):R_{-\alpha}(q_2)),
\end{equation}
where $F_{-\alpha}$ is the convex conjugate $F_\alpha^*$ of $F_\alpha$ and $R_{-\alpha}$ is the dual representation~\cite{nielsen2009dual}.
\end{Proposition}

Amari~\cite{amari2009alpha} proved that the intersection of the class of $f$-divergences extended on positive measures  with the class of Bregman divergences  are exactly the extended $\alpha$-divergences.

\section{Application: Intersection of $\alpha$-divergence spheres}\label{sec:inter}

It is well-known that the intersection of two $m$-dimensional Euclidean spheres $\Sigma_1$ and $\Sigma_2$ is yet another $(m-1)$-dimensional Euclidean sphere.
This section first deals with the intersection of two Bregman spheres induced by the same Bregman generator yielding another type of $(d-1)$-dimensional Bregman sphere (i.e., a sub-dimensional Bregman divergence sphere).

Let $\sigma_F(\theta,r)$ and $\sigma^\drev_F(\theta,r)$ be the right and left Bregman spheres, respectively:
\begin{eqnarray*}
\sigma_F(\theta,r) := \left\{\theta'\in\Theta \st B_F(\theta':\theta) = r\right\},\quad
\sigma_F^\drev(\theta,r) := \left\{\theta'\in\Theta \st B_F(\theta:\theta') = r\right\},
\end{eqnarray*}
where $\drev$ denotes the reference duality~\cite{zhang2004divergence}.
Using the Legendre-Fenchel convex duality ($\eta=\nabla F(\theta)=:\theta^*$ and $\theta=\nabla F^*(\eta)=:\eta^*$), we have~\cite{boissonnat2010bregman}:
\begin{eqnarray*}
\sigma_F^\drev(\theta,r) = \left(\sigma_{F^*}(\eta=\theta^*,r)\right)^*,\quad
\sigma_{F^*}^\drev(\eta,r) = \left(\sigma_{F}(\eta^*,r)\right)^*.
\end{eqnarray*} 

We now explain the space of Bregman spheres~\cite{boissonnat2010bregman}:
For a parameter $\theta$, let $\hat\theta:=(\theta,F(\theta))$ be the vertical lifting of $\theta$ onto the epigraph $\calF:=\{(\theta,F(\theta)\st\theta\in\Theta)\}$ and $\downarrow\hat\theta=\theta$ be the vertical projection so that $\downarrow(\hat\theta)=\theta$.
Let $\Sigma=\sigma_F(\theta,r)$ be a right Bregman sphere.
The pointwise lifted right Bregman sphere $\hat\Sigma=\hat\sigma_F(\theta,r)$ is supported by a non-vertical hyperplane $H_{\Sigma}$ of equation~\cite{boissonnat2010bregman}:
$$
H_{\Sigma}=H_{\theta,r}: y=\inner{\theta'-\theta}{\nabla F(\theta)}+F(\theta)+r.
$$
That is, we have $\Sigma = \downarrow(H_\Sigma\cap\calF)$.

Conversely, the intersection of any non-vertical hyperplane $H: y=\inner{\theta_a}{\theta}+b$ with $\calF$ projects vertically on $\Theta$ 
as a Bregman sphere with center
$\theta_H=\nabla F^{-1}(\theta_a)$ and radius $r_H=\inner{\theta_a}{\theta_H}-F(\theta_H)+b$.

The supporting hyperplane of the right Bregman sphere (Figure~\ref{fig:rightleftsphere}, blue) is obtained by lifting the center $\theta$ onto the epigraph, taking the tangent plane at $(\theta,F(\theta))$ and translating that tangent plane vertically by $r\geq 0$.
The left Bregman sphere  (Figure~\ref{fig:rightleftsphere}, red) is obtained by first illuminating~\cite{TDA-Bregman-2017} from $(\theta,F(\theta)-r)$ the epigraph $\calF$, and then projecting vertically the illuminated portion of $\calF$ onto $\Theta$.
Let $\calF^\circ$ denote the open interior of the epigraph $\calF$.
Then the illuminated portion of the graph $\partial\calF=\{(\theta,F(\theta)) \st \theta\in\Theta\}$ is given by 
$
\left(\barCH
\left( \{ (\theta,F(\theta)-r) \} \cup \calF \right)
\backslash \calF^\circ \right) \cap\partial\calF$,
where $\barCH$ denotes the closed convex hull operator.
Figure~\ref{fig:ExampleBregmanSpheres} illustrates the left and right Bregman sphere obtained for the scalar Shannon negentropy function
 $F_S(\theta)=\theta\log\theta$.

\begin{figure}
\centering
\includegraphics[width=\textwidth]{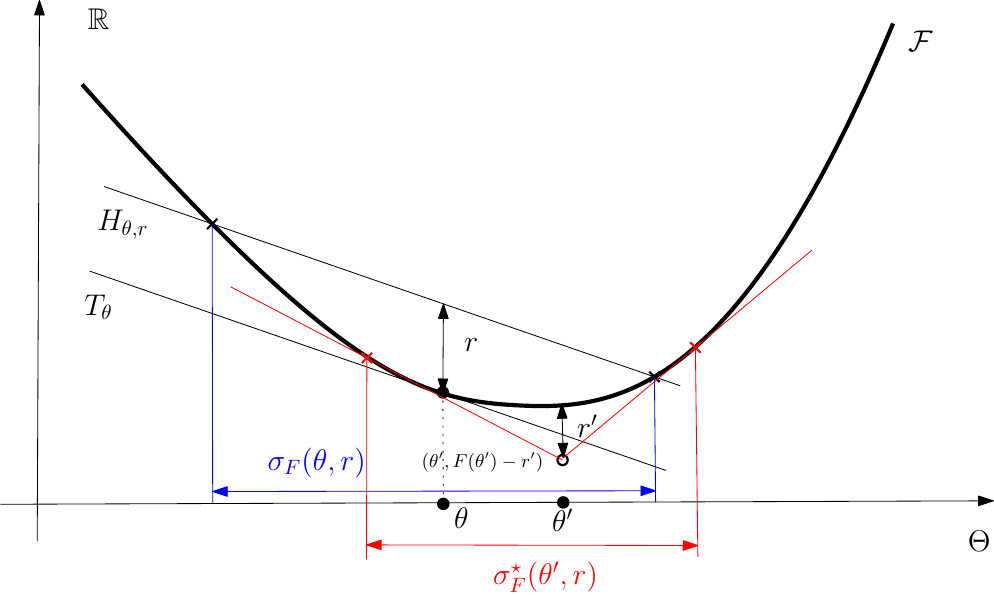}

\caption{Construction of Bregman spheres from the Bregman generator epigraph: Right Bregman sphere $\sigma_F(\theta,r)$ (blue) and 
left Bregman sphere $\sigma_F^\star(\theta',r')$ (red).}\label{fig:rightleftsphere}
\end{figure}


\def\wfig{0.32\textwidth}

\begin{figure}\centering
\begin{tabular}{ccc}
\includegraphics[width=\wfig]{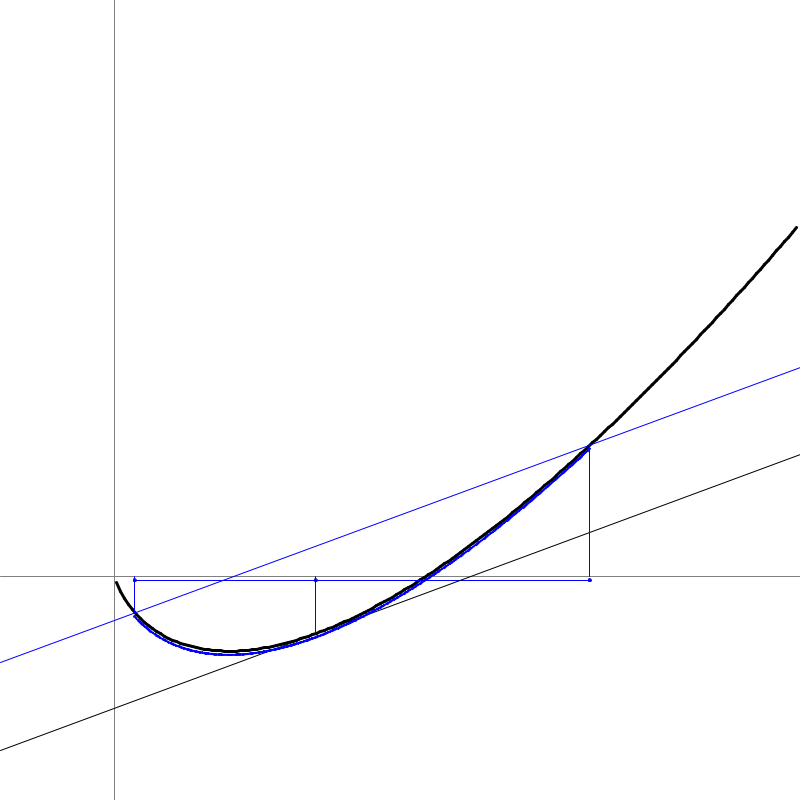} &
\includegraphics[width=\wfig]{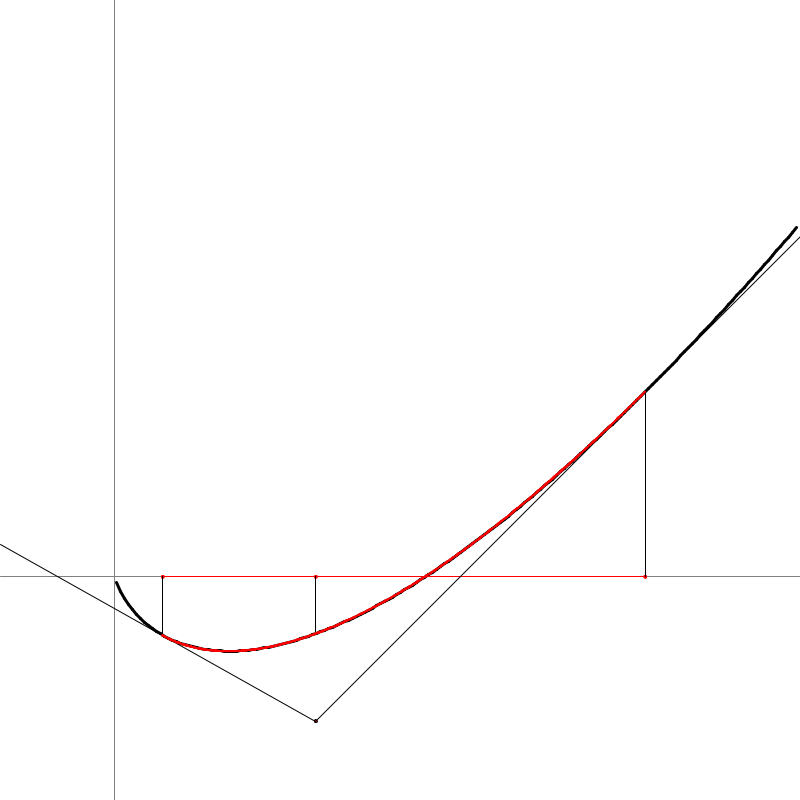} &
\includegraphics[width=\wfig]{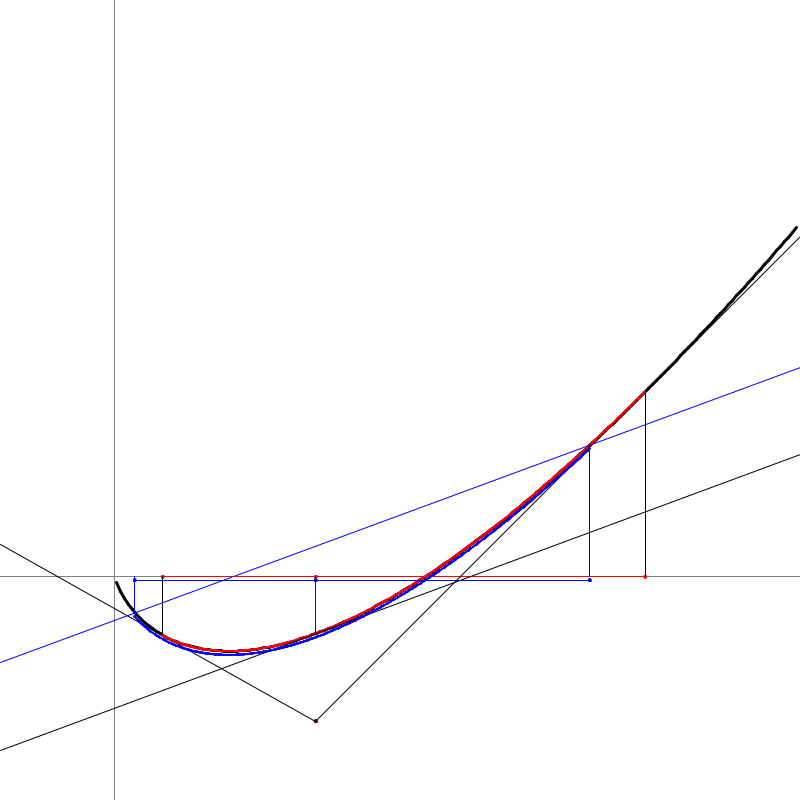} 
\end{tabular}
\caption{Illustrations of the lifting of left and right Bregman spheres on the potential function induced by the  Shannon neg-entropy function. 
Right Bregman sphere is obtained as the vertical projection of the hyperplane with the potential function (left figure). 
Left Bregman sphere is given as the intersection of the graph of the potential function with the convex hull including the point (i.e., the illuminated part of the potential function from that point, middle figure). Superposition of both right and left Bregman spheres (right figure).}

\label{fig:ExampleBregmanSpheres}
\end{figure}

The lifting of Bregman spheres $\sigma_F$ to the potential epigraph $\calF$ allows one to build efficient algorithms for computing the intersection of Bregman spheres:
 
\begin{Proposition}[\cite{boissonnat2010bregman}]
The intersection of two Bregman spheres $\Sigma_1$ and $\Sigma_2$ of dimension $m$ is a sub-dimensional Bregman sphere $\Sigma_{12}$ of dimension $m-1$.
\end{Proposition}

\begin{proof}
Let $\Sigma_1=\sigma_F(\theta_1,r_1)=\{\theta' \st B_F(\theta':\theta_1)= r_1\}$ 
and $\Sigma_2=\sigma_F(\theta_2,r_2)=\{\theta' \st B_F(\theta':\theta_2)= r_2\}$ 
be two right Bregman spheres.
We have $\Sigma_1=\downarrow(H_{\Sigma_1}\cap F)$ and 
$\Sigma_2=\downarrow(H_{\Sigma_2}\cap F)$.
Thus, we have 
\begin{eqnarray*}
\Sigma_1\cap\Sigma_2 &=& \downarrow(H_{\Sigma_1}\cap F) \cap \downarrow(H_{\Sigma_1}\cap F) = \downarrow( (H_{\Sigma_1}\cap H_{\Sigma_2} )\cap F). 
\end{eqnarray*}
Let $H_{12}=H_{\Sigma_1}\cap H_{\Sigma_2}$ be a flat (affine subspace) of dimension $m-1$, and let $H_{12}^|$ be the vertical flat (with respect to $y$-axis) of $\bbR^{m+1}$ passing through $H_{12}$. Then $\calF_{12}=\calF\cap H_{12}^|$ is the epigraph of a Bregman generator $F_{12}$.
Thus the intersection of $\Sigma_1$ with $\Sigma_2$ is a Bregman sphere with respect to generator $F_{12}$.
\end{proof}

Figure~\ref{fig:paraboloidspheres} illustrates the representation of 2D Euclidean spheres by 3D planes cutting the paraboloid of revolution.
The boolean algebra of the union and intersection of Bregman balls can be manipulated as intersections of halfspaces cutting the epigraph of the potential function~\cite{boissonnat2010bregman}.

\begin{figure}
\centering
\begin{tabular}{cc}
\includegraphics[width=0.45\textwidth]{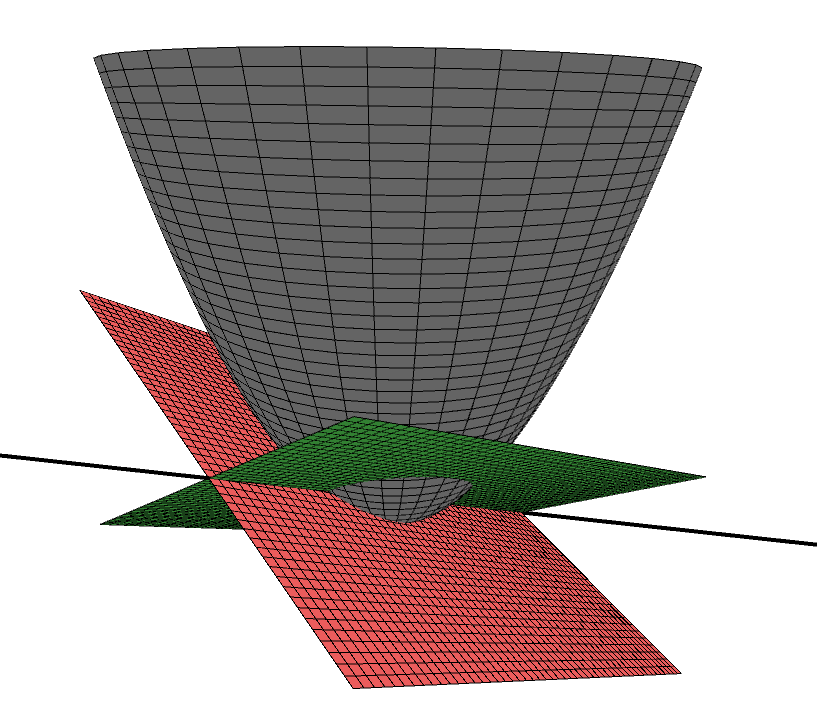} &
\includegraphics[width=0.45\textwidth]{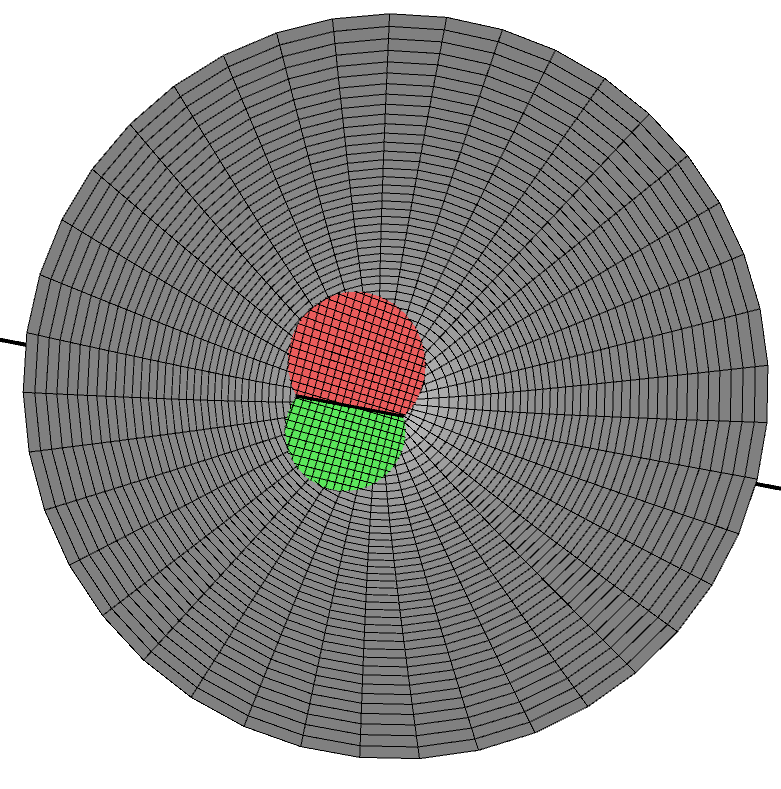}
\end{tabular}

\caption{Representing Bregman spheres $\Sigma$ as intersections of  hyperplanes $H_\Sigma$ with a potential function $\calF$ (paraboloid of revolution for Euclidean spheres). Left: Two 2D Bregman spheres represented by 3D planes $H_{\Sigma_1}$ and $H_{\Sigma_2}$. 
Right: Their union $\Sigma_1\cup\Sigma_2$ can be visualized by looking at the interior of the paraboloid from the top.
Their radical axis is the vertical projection of the intersection $H_{\Sigma_1}\cap H_{\Sigma_2}$.\label{fig:paraboloidspheres} }
\end{figure}

Now consider $\alpha$-divergence spheres ($\alpha$-spheres for short):
Let $\sigma_\alpha(\theta,r) := \left\{\theta'\in\Theta \st D_\alpha(\theta':\theta) = r\right\}$ be the right $\alpha$-sphere of center $\theta$ and radius $r$.
Then the left $\alpha$-divergence sphere $\sigma_\alpha^\drev(\theta,r):=\left\{\theta'\in\Theta \st D_\alpha(\theta:\theta') = r\right\}=\sigma_{-\alpha}(\theta,r)$.
Since $\alpha$-divergences are curved representational Bregman divergences, we may consider the following algorithm to calculate the intersection of $n$ $\alpha$-spheres $\Sigma_1=\sigma_\alpha(\theta_1,r_1),\ldots,  \Sigma_n=\sigma_\alpha(\theta_n,r_n)$:
Let $S_i=\sigma_{F_\alpha}(r_\alpha(\theta_i),r_i)$ be the representational Bregman sphere corresponding to $\Sigma_i$ for $i\in\{1,\ldots,n\}$, and $\calF_\alpha$ the epigraph of $F_\alpha(\theta)=\frac{2}{1+\alpha}\sum_{i=1}^m \left(\frac{1-\alpha}{2} \theta_i\right)^{\frac{2}{1-\alpha}}$.

\begin{Proposition}[Intersection of $\alpha$-spheres]\label{prop:intersection}
We have 
\begin{equation}
\cap_{i=1}^n \sigma_\alpha(\theta_i,r_i) = \left(\downarrow\left( \left(\cap_{i=1}^n H_{S_i} \right)\cap \calF_\alpha\right)\right) \cap \alphasimplex,
\end{equation}
where $\alphasimplex:=\{R_\alpha(x) \st x\in\Delta_m\}$ is the $\alpha$-representation of the probability simplex.
\end{Proposition}

The representational Bregman divergence point of view of $\alpha$-divergences (Proposition~\ref{prop:alphabdrep}) allows one to extend easily the computational geometry toolbox of Bregman divergences to $\alpha$-divergences. For example, we may consider nearest neighbor query data structures for $\alpha$-divergences following~\cite{nielsen2009bregman}, or consider the smallest enclosing $\alpha$-divergence ball following~\cite{nielsen2008smallest}.
The circumcenter of the  smallest enclosing ball is also called the Chebyshev point~\cite{candan2020chebyshev} and finds application in universal coding in information theory.

\section{Summary}

In this work, we defined the notion of curved Bregman divergences (Definition~\ref{def:curvedEF}) by analogy to curved exponential families~\cite{CEF-Amari-1982}.
We proved that the barycenter of a finite weighted set of points belonging to a nonlinear subspace $\calU$ (submanifold) of a Bregman parameter space $\Theta$ amounts to a right Bregman projection of the center of mass onto $\calU$ (Theorem~\ref{thm:proj}). 
We reported several examples of curved Bregman divergences: (1) symmetrized Bregman divergences (Jeffreys-Bregman divergence~\cite{nielsen2019jensen}), 
  (2) Kullback-Leibler divergence between circular complex normal distributions, and
	(3) Kullback-Leibler divergence between categorical distributions (Example~\ref{ex:curvedKLDcat}).
	As an application, we handled $\alpha$-divergences as curved representational Bregman divergences and show how to compute the intersection of $\alpha$-spheres 
	(Proposition~\ref{prop:alphabdrep}), i.e., spheres defined according to $\alpha$-divergences.
	
Future work will consider the information geometry of submanifolds of the Hessian manifolds with applications.	
	
\bibliographystyle{plain}
\bibliography{subdimcurvedBDV5}

\appendix

\section*{Code snippet in {\sc Maxima}}\label{sec:maxima}

\subsection{Circle curved Euclidean divergence yields cosine dissimilarity}\label{sec:a:cosdis}
The following code snippet is written in the computer algebra system 
{\sc Maxima} (\url{https://maxima.sourceforge.io/}):

\begin{verbatim}
/* curved circle Euclidean (Bregman) divergence = cosine dissimilarity */
kill(all)$
BD(theta1,theta2):=(1/2)* ( (theta1[1]-theta2[1])**2  + (theta1[2]-theta2[2])**2);
theta(u):=[cos(u),sin(u)];
exu1:random(float(2*%pi));
exu2:random(float(2*%pi));
extheta1: theta(exu1);
extheta2: theta(exu2);
BD(extheta1,extheta2);
BDc(u1,u2):=1-cos(u1-u2);
ratsimp(%);
BDc(exu1,exu2);
\end{verbatim}

\subsection{Geometric mean is Itakura-Saito symmetrized/COSH centroid}\label{sec:a:sbdis}

\begin{verbatim}
/* Itakura-Saito divergence: Check that geometric mean is symmetrized centroid thetaS */
kill(all);
F(theta):=-log(theta);
gradF(theta):=-1/theta;
BD(theta1,theta2):=F(theta1)-F(theta2)-(theta1-theta2)*gradF(theta2);
thetabar  : (theta1+theta2)/2;  /* arithmetic mean */
thetaubar: 2/((1/theta1)+(1/theta2)); /* harmonic mean */
thetaS : sqrt(theta1*theta2); /* geometric mean */
/* let us check that thetaS belongs to the mixed bisector */
BD(thetabar,thetaS)-BD(thetaS,thetaubar);
ratsimp(%);
%,logcontract;
\end{verbatim}

\end{document}